\documentclass[10pt,conference,a4paper]{IEEEtran} %

\IEEEoverridecommandlockouts

\usepackage{cite}
\usepackage[hidelinks]{hyperref}
\usepackage{amsmath,amssymb,amsfonts}
\usepackage{algpseudocode} 
\usepackage{graphicx}
\usepackage{textcomp}
\usepackage{xcolor}
\usepackage{accents}
\usepackage{tikz}
\usetikzlibrary{calc}
\usepackage{amsthm}

\def\BibTeX{{\rm B\kern-.05em{\sc i\kern-.025em b}\kern-.08em
    T\kern-.1667em\lower.7ex\hbox{E}\kern-.125emX}}

\newcommand{\NN}{\mathbb{N}}
\newcommand{\HdTbcBW}{HdTbcBW}
\newtheorem{definition}{Definition}    
\newtheorem{lemma}{Lemma}    
\newtheorem{theorem}{Theorem}
\newtheorem{corollary}{Corollary}
\newtheorem{proposition}{Proposition}    
    
\begin{document}

\title{Rerailing Automata}

\author{\IEEEauthorblockN{Rüdiger Ehlers}%
\IEEEauthorblockA{
  Institute for Software and Systems Engineering \\ Clausthal University of Technology \\
  Clausthal-Zellerfeld \\
  Germany
}}

\maketitle
\thispagestyle{plain} %
\pagestyle{plain} %

\begin{abstract}

In this paper, we introduce rerailing automata for $\omega$-regular languages. They generalize both deterministic parity (DPW) and minimized history-deterministic co-Büchi automata (with transition based acceptance, \HdTbcBW{}) while  combining their favorable properties. 
In particular, rerailing automata can represent arbitrary $\omega$-regular languages while allowing for polynomial-time minimization, just as \HdTbcBW{} do. Since DPW are a special case of rerailing automata, a minimized rerailing automaton is never larger than the smallest deterministic parity automaton for the same language. 
We also show that rerailing automata can be used as a replacement for deterministic parity automata for the realizability check of open systems.

The price to be paid to obtain the useful properties of rerailing automata is that the acceptance condition in such automata refers to the \emph{dominating} colors along \emph{all} runs for a given word,
where just as in parity automata, the dominating color along a run is the lowest one occurring infinitely often along it. A rerailing automaton accepts those words for which the \emph{greatest} of the dominating colors along the runs is even.
Additionally, rerailing automata guarantee that every prefix of a run for a word can be extended to eventually reach a point from which all runs for the word extending the prefix have the same dominating color, and it is even if and only if the word is in the language of the automaton. We show that these properties together allow
 characterizing the role of each state in such an automaton in a way that relates it to state combinations in a sequence of co-Büchi automata for the represented language. This characterization forms the basis of the polynomial-time minimization approach in this paper.

\end{abstract}

\section{Introduction}

Automata over infinite words are a useful tool for the synthesis and analysis of reactive systems. In contrast to temporal logic, which is suitable for expressing properties that reactive systems under concern should have, automata serve as a more technical representation of such properties and are used as (intermediate) specification formalism in model checking and reactive synthesis procedures.

These procedures typically have computation times that increase with the sizes of the automata used as input, which leads to an interest in keeping these automata \emph{small}, typically measured in the number of states. For automata over finite words, the complexity of minimizing them is well-known. The Myhill-Nerode theorem \cite{nerode1958linear} provides a way to minimize deterministic such automata in polynomial time, while for non-deterministic automata, the decision version of the minimization problem (i.e., determining whether a smaller automaton exists) is PSPACE-complete \cite{DBLP:journals/siamcomp/JiangR93}. 

For automata over infinite words, the landscape is a bit more complex. 
For non-deterministic Büchi automata, which are useful for model checking, PSPACE-hardness of minimality checking is known. However, many verification and synthesis procedures require deterministic automata, for which the complexity of minimality checking is much more nuanced.
For instance, it is known that deterministic weak automata can be minimized in polynomial time \cite{DBLP:journals/ipl/Loding01}. Such automata are however  heavily restricted in their expressivity. Even for the slightly more expressive class of deterministic Büchi automata, the minimality problem is already NP-hard for state-based acceptance \cite{DBLP:conf/fsttcs/Schewe10}. 

On the positive side, it has been shown that by extending deterministic branching slightly to \emph{history-de\-ter\-minis\-tic} branching, and combining it with edge-based acceptance rather than state-based acceptance, 
we obtain an automaton model for co-Büchi languages with a polynomial-time minimization algorithm \cite{DBLP:journals/lmcs/RadiK22}. 
This slight extension of the branching condition is attractive because the resulting \emph{history-de\-ter\-minis\-tic transition-based co-Büchi word automata} (\HdTbcBW) can be used as drop-in replacement for deterministic co-Büchi automata in a number of applications in which traditionally, deterministic automata are used, such as for the quantitative verification of Markov decision processes \cite{DBLP:conf/lata/KleinMBK14} and reactive synthesis \cite{DBLP:conf/csl/HenzingerP06,michaud2018reactive}.

The polynomial-time minimization approach for history-deterministic co-Büchi automata by Abu Radi and Kupferman~\cite{DBLP:journals/lmcs/RadiK22} exploits the properties of the co-Büchi acceptance condition and hence does not generalize beyond co-Büchi languages. However, it can be observed that the minimized history-deterministic automata computed by Abu Radi's and Kupferman's approach have specific properties beyond history-determinism. In particular, the minimized automata are \emph{language-deterministic} and deterministic in the accepting transitions.
Given that it is known that minimizing history-deterministic Büchi or parity automata is NP-hard \cite{DBLP:conf/fsttcs/Schewe20} (for state-based acceptance), so history-determinism alone is not what made the positive result by Abu Radi and Kupferman possible, perhaps minimal automata having these specific properties that go beyond history-determinism are the key to obtaining a polynomial-time minimizable automaton model for $\omega-$regular languages? 

More precisely, we ask the following question: \emph{Is there an automaton class that allows polynomial-time minimization and that is a strict extension of both deterministic parity automata and minimized history-deterministic co-Büchi automata with transition-based acceptance?} 
Deterministic parity automata being a special case would allow to encode any omega-regular property, while extending  \emph{minimized} \HdTbcBW{} is motivated by keeping its favourable properties that enable polynomial-time minimization without losing the applicability of these automata in practice.

In this paper, we answer this question positively and present \emph{rerailing automata}. 
Rerailing automata are obtained by interpreting the properties of minimized \HdTbcBW{} in a novel way that generalizes beyond co-Büchi languages. Starting from encoding the co-Büchi acceptance condition as a min-even parity acceptance condition over the colors $\{1,2\}$, we make the following reinterpretations:
\begin{itemize}
\item 
Accepted words in \HdTbcBW{} have a run with dominating color $2$, and rejected words only have runs with dominating color $1$. We can reformulate this property as defining that the automaton accepts the words for which the \emph{greatest} of the dominating colors along its runs is even.
\item In minimized \HdTbcBW{}, all rejected words have runs with color $1$ as the dominating color, and all run prefixes of runs for accepted words can be extended to a run that eventually only takes \emph{deterministic} transitions with color $2$.
We can reinterpret this property as defining that for all words, all prefix runs $\pi$ can be extended to reach a state so that afterwards, regardless of which transitions are taken from there, the resulting run always recognizes the word with the same dominating color, and it is even if and only if the word is in the language represented by the automaton. Furthermore, by rerouting the run $\pi$, the dominating color of the run can only increase -- in the case of \HdTbcBW{}, it can increase from 1 to 2.
\end{itemize}
In rerailing automata, we combine both of these reinterpretations of the properties of minimized \HdTbcBW{} with supporting colors beyond $\{1,2\}$ in a parity acceptance condition. 

We show in this paper that rerailing automata both generalize from deterministic parity automata and minimized \HdTbcBW{} and how they can be minimized in polynomial time. As they are special cases of deterministic parity automata, for every language, they never need to be bigger than the smallest deterministic parity automaton for the same language.  
With this property, they are a suitable drop-in replacement of deterministic parity automata in some applications, and we show how rerailing automata can be used for reactive synthesis via a reduction to parity game solving, where the reduction is only a minor modification of the corresponding reduction for deterministic parity automata.

\section{Preliminaries}

\textbf{Words:} Given a finite set $\Sigma$ (an \emph{alphabet}), we denote the set of finite sequences of $\Sigma$ as $\Sigma^*$ and the set of infinite sequences as $\Sigma^\omega$. A set of words is also called a \emph{language}. For a given alphabet, we also call the set of words over the alphabet the \emph{universal language}.  The finite word of length $0$ is also denoted as $\epsilon$. We say that some finite word $w = w_0 \ldots w_n$ \label{def:appears}\emph{appears} at some position $i$ in some other (finite or infinite) word $w' = w'_0 \ldots$ if we have $w'_i w'_{i+1} \ldots w'_{i+n} = w_0 \ldots w_n$. Whenever such a position exists, we also say that $w$ is a \label{def:consecutiveSubword}\emph{consecutive subword} of $w'$.

\textbf{Automata over infinite words (with edge-based acceptance):} An \emph{automaton structure} is a tuple $\mathcal{A} = (Q,\Sigma,\delta,q_0)$ with the finite set of states $Q$, the alphabet $\Sigma$, the transition relation $\delta \subseteq Q \times \Sigma \times Q \times \NN$, and the initial state $q_0 \in Q$. 

Give an infinite word $w = w_0 w_1 \ldots \in \Sigma^\omega$, we say that a sequence $\pi = \pi_0 \pi_1 \ldots \in Q^\omega$ is a \emph{run} of $w$ together with a \emph{color sequence} $\rho = \rho_0 \rho_1 \ldots \in \NN^\omega$ if we have $\pi_0 = q_0$ and for all $i \in \NN$, we have $(\pi_i,w_i,\pi_{i+1},\rho_i) \in \delta$. We say that $\pi/\rho$ has dominating color $\mathsf{dominatingColor}(\rho) = c'$ if the lowest color appearing infinitely often in $\rho$ is $c'$.
We only consider automaton structures in this paper for which every run has only one corresponding color sequence, so that it makes sense to refer to the dominating color of the run itself.
We denote by $\mathsf{Runs}(\mathcal{A},w)$ the set of run/color sequence combinations of $\mathcal{A}$ for $w$.

An automaton is an automaton structure together with a combined branching/acceptance condition. 
If $\mathcal{A}$ is declared to be a \emph{deterministic} or \emph{non-deterministic} automaton, the automaton accepts a word if there exists a run with even dominating color (\emph{min-even acceptance}). We call such a run \emph{accepting}. For a deterministic automaton, we furthermore have that every word has exactly one run.
We say that a non-deterministic automaton $\mathcal{A}$ is (also) \emph{history-deterministic} if there exists some \emph{strategy} function $f: \Sigma^* \rightarrow Q$ such that if a word $w = w_0 w_1 \in \Sigma^\omega$ has a run with even dominating color, then the unique run $\pi = \pi_0 \pi_1 \ldots$ with $\pi_i = f(w_0 \ldots w_{i})$ for every $i \in \NN$ is an accepting run. 
For non-deterministic or deterministic automata, the automaton accepts those words that have an accepting run.
We denote by $\mathcal{L}(\mathcal{A})$ the \emph{language} of the automaton, i.e., the set of words that it accepts. If the same automaton structure is interpreted with multiple different branching/acceptance conditions, we will attribute the $\mathcal{L}$ symbol with a subscript denoting which one we refer to.

We only consider \emph{complete} automata in this paper, i.e., so that for each $q \in Q$ and $x \in \Sigma$, there is some transition $(q,x,q',c) \in \delta$ (for some $q'$ and $c$). We denote by $\mathcal{A}_q$ for some $q \in Q$ the same automaton as $\mathcal{A}$, but with the initial state replaced by $q$.

Deterministic or non-deterministic automata of the form defined above are also called \emph{parity automata}. Parity automata in which only the colors $1$ and $2$ are used are also called \emph{co-Büchi automata}.
Furthermore, history-deterministic parity or co-Büchi automata are also called good-for-games parity/co-Büchi automata in the literature, because the notions of history-determinism and good-for-games acceptance coincide for them \cite{DBLP:conf/fsttcs/BokerL21}.

\textbf{Analyzing languages:} 
Let $L$ be a language. The set of \emph{residual} languages is defined as $\mathsf{ResidualLanguages}(L) = \{ L' \subseteq \Sigma^\omega \mid \exists w \in \Sigma^*. L' = \mathsf{Residual}(L,w) \}$, where $\mathsf{Residual}(L,w) = \{\tilde w \in \Sigma^\omega \mid w \tilde w \in L\}$ is the residual language of $L$ for $w$.
We say that some tuple $R^L = (S^L,\Sigma,\delta^L,s^L_0)$ is a \emph{residual language tracking automaton} (\emph{RLTA}) with the finite set of states $S^L$, the deterministic transition function $\delta : S^L \times \Sigma \rightarrow S^L$ and the initial state $s^L_0$ if there exists some \emph{bijective} mapping function $\tilde f : S^L \rightarrow \mathsf{ResidualLanguages}(L)$ such that $\tilde f(s^L_0)=L$ and for every $s \in S^L$ and $x \in \Sigma$, we have $\tilde f(\delta^L(s,x)) = \mathsf{Residual}(\tilde f(s),x)$.

Every (parity) automaton $\mathcal{A} = (Q,\Sigma,\delta,q_0)$ for some language $L$ has a finite set of residual languages. We say that the automaton is \emph{language-deterministic} if there exists a mapping $f : Q \rightarrow \mathsf{ResidualLanguages}(L)$ such that $f(q_0) = L$ and for each $(q,x,q',c) \in \delta$, we have $f(q') = \{w \in \Sigma^\omega \mid xw \in f(q) \}$.
Note that unlike for residual language tracking automata, we do not expect the mapping to be bijective, however.

\textbf{Analyzing automata:}
Given some automaton structure $\mathcal{A} = (Q,\Sigma,\delta,q_0)$, we say that a tuple $(Q',\delta')$ with $Q \subseteq Q'$ and $\delta' \subseteq \delta$ is a strongly connected component (SCC) of $\mathcal{A}$ if there exists some word/run combination such that $Q'$ are exactly the states visited infinitely often along the run, and $\delta'$ are exactly the transitions taken infinitely often along it. We say that the SCC is \emph{maximal} if it cannot be extended by any state or transition without the pair losing the property that it is an SCC (for some word). We say that a run $\pi$ for some word $w = w_0 w_1 \ldots \in \Sigma^\omega$ \emph{gets stuck} in some SCC in $(Q',\delta')$ if along $\pi$ only finitely many transitions outside $\delta'$ are taken. Such runs then consist of a \emph{prefix} run $\pi_0 \ldots \pi_n$ for the \emph{prefix word} $w_0 \ldots w_{n-1}$ (for some $n \in \NN$), followed by the parts of the word and run in which the run part only takes transitions in $\delta'$.

Note that (non-deterministic) co-Büchi automata accept a word if and only if there is a run that gets stuck in some SCC consisting only of accepting transitions (i.e., those with color $2$, and we call such SCCs \emph{accepting}). We will consider co-Büchi automata that are \emph{deterministic in the accepting transitions}, i.e., such that for every $q \in Q$ and $x \in \Sigma$, if there is a transition $(q,x,q',2) \in \delta$, then this is the only transition from $q$ for $x$. By slight abuse of notation, we also treat $\delta$ as a function $Q \times \Sigma \rightarrow Q$ mapping states and characters to the only successor state when a suitable deterministic transition is present, and $\delta(q,x)$ is undefined otherwise. We furthermore extend $\delta$ to sequences of characters, so that $\delta(q,w_0 \ldots w_n) = \delta(\delta(q,w_0 \ldots w_{n-1}),w_n)$ for words $w_0 \ldots w_n$ of length greater $0$, and $\delta(q,\epsilon)=q$. We also use this generalization for the transition function in residual language tracking automata.
We say that some finite word $w$ is \emph{safely accepted} from some state $q \in Q$ if $\delta(q,w)$ is defined. 
To be able to reason about which states are reachable for which finite words, for some state $q$ and finite word $w$, we define $\delta^+(q,w)$ to be the set of states having prefix runs from $q$ under the finite word $w$.

\textbf{Parity games:}
A parity game (see, e.g., \cite{DBLP:reference/mc/BloemCJ18}) is a tuple $\mathcal{G} = (V^0,V^1,E^0,E^1,C,v^0)$ with the disjoint finite set of \emph{vertices} $V^0$ and $V^1$ of the two players, the edges $E^0 \subseteq V^0 \times (V^0 \cup V^1)$ and $E^1 \subseteq V^1 \times (V^0 \cup V^1)$ of the two players, a coloring function $C : V^0 \cup V^1 \rightarrow \NN$, and an initial vertex $v^0$.

The interaction of the two players (the even player with number $0$ and the odd player with number $1$) is captured in a \emph{play} $\pi = \pi_0 \pi_1 \pi_2 \ldots \in (V^0 \cup V^1)^\omega$ of the game, where $\pi_0 = v^0$, and for each position $i \in \NN$, it is the role of player $j \in \{0,1\}$ to choose an edge with $\pi_i$ as first element if and only if $\pi_i \in V^j$. With such an edge $(v,v') \in E^j$ with $v = \pi_i$, we then have $\pi_{i+1} = v'$. The (dominating) color of the play, denoted as $\mathsf{Color}(\pi)$, is the lowest number appearing infinitely often in $C(\pi_0) C(\pi_1) C(\pi_2) \ldots$.

We say that player $0$ wins the play if $\mathsf{Color}(\pi)$ is even, and otherwise player $1$ wins the play. It has been shown in the literature that either player $0$ or player $1$ have a strategy to win the game, i.e., to ensure that a play that is the outcome of the interaction by the two players is winning for the respective player. If the number of different colors of the vertices is limited to some constant, determining which player \emph{wins} the game, i.e., has a strategy to win it, can be performed in time polynomial in the size of the game~\cite{DBLP:reference/mc/BloemCJ18,DBLP:journals/jcss/Schewe17}.

\section{Introducing rerailing automata}
\label{sec:IntroductionRerailingAutomata}
In this section, we motivate and define rerailing automata.
We start with a description of the properties of minimized history-deterministic co-Büchi automata with transition-based acceptance (\HdTbcBW{}) that we generalize from, then introduce rerailing automata formally, and finally prove some properties of rerailing automata that are used later in this paper.

\subsection{Analyzing the properties of minimal history-deterministic co-Büchi automata}
\label{sec:minimalHdTbcBWProperties} %

Abu Radi and Kupferman  \cite{DBLP:journals/lmcs/RadiK22} gave a procedure for minimizing \HdTbcBW{} and making them canonical in polynomial time. We review the properties that the resulting automata have.

Let $\mathcal{A} = (Q,\Sigma,\delta,q_0)$ be a minimal and canonical \HdTbcBW{}. It has the following properties:%
\begin{enumerate}
\item The automaton $\mathcal{A}$ is language-deterministic.
\item Whenever $\mathcal{A}$ takes a rejecting transition to a state $q'$ along a run, then there are rejecting transitions to all other states with the same residual language as $q'$ as well. 
More formally, for a labeling $m$ from states to residual languages and for every transition $(q,x,q',c) \in \delta$ with $c=1$ (i.e., that is \emph{rejecting}), we have that for all states $q'' \in Q$ with $m(q')=m(q'')$, there is also a transition $(q,x,q'',c)$ in $\delta$.
\item The automaton $\mathcal{A}$ is deterministic in the accepting transitions.
\item All maximal accepting SCCs in $\mathcal{A}$ are minimized, i.e., in every SCC, there are no two states with the same sets of safely accepted words and the same residual languages.
\item For any two states $q$, $q'$ with the same residual language that are in different maximal accepting SCCs, we have that the sets of \emph{safely accepted} words from $q$ and $q'$ are incomparable. 
\item For every residual language that has no state that is part of an accepting SCC, there is exactly one state in the automaton with that residual language.
\end{enumerate}
Figure~\ref{fig:exampleCoBuchiGFG} shows an automaton that exemplifies these properties for a relatively simple language. When stated in form of an \emph{$\omega$-regular expression}, the language over the alphabet $\Sigma = \{a,b,c\}$ of the automaton is $L = \Sigma (\Sigma^2)^* (\{a,b\}^\omega \allowbreak{} \cup \allowbreak{} (\{b,c\}\{a,b\} \cup ab)^\omega)$, i.e., exactly those words are in $L$ for which starting at an \emph{odd} position in the word, the word ends with either (1) $\{a,b\}^\omega$ or (2) $(\{b,c\}\{a,b\} \cup ab)^\omega$. The two accepting maximal strongly connected components of the automaton each take care of one of the (1) or (2) cases. They are connected by rejecting transitions, so any accepted word has a run that eventually gets ``stuck'' in one of these.

The automaton tracks the residual language, for which in this case it suffices to track whether an even or odd number of letters have been read by the automaton so far. In this example, it is easy to see that the automaton does so, as every transition from a state $q_i$ with an even $i$ leads to a state $q_{j}$ with an odd $j$ and vice versa. 

If can be observed that the accepting transitions are deterministic. 
Whenever for a state $q$ and a letter $x$, there is no outgoing accepting transition for $x$ from $q$, there are rejecting transitions to all states with the corresponding residual language. For instance, for $q_1$ and $c$, there are transitions to $q_0$, $q_2$, and $q_4$.  

\looseness-1 When comparing the \emph{safe languages} of the states, we can observe that state pairs across different maximal accepting SCCs have incomparable safe languages or incomparable residual languages. Within an SCC, we only know that no two states have the same set of safely accepted words (the \emph{safe language}) and the same residual languages. Here, the safe language of $q_2$ is a superset of the safe language of $q_4$, but because the superset relation is strict, the states $q_2$ and $q_4$ cannot be merged.

\begin{figure}
\begin{tikzpicture}
\draw (0,0) node[draw,shape=circle,minimum size=0.7cm,thick] (q0) {$q_0$};
\draw (2,0) node[draw,shape=circle,minimum size=0.7cm,thick] (q1) {$q_1$};
\draw[->,thick] (q0) to[bend left=10] node[above] {$a,b$} (q1);
\draw[->,thick] (q1) to[bend left=10] node[below] {$a,b$} (q0);

\draw[fill] (-0.3,0.7) coordinate (init) circle (0.05cm);
\draw[->,thick] (init) -- (q0);

\draw (5,0) node[draw,shape=circle,minimum size=0.7cm,thick] (q2) {$q_2$};
\draw (7,0) node[draw,shape=circle,minimum size=0.7cm,thick] (q3) {$q_3$};
\draw (5,-2) node[draw,shape=circle,minimum size=0.7cm,thick] (q4) {$q_4$};

\draw[->,thick] (q2) to[bend left=10] node[above] {$a,b$} (q3);
\draw[->,thick] (q3) to[bend left=10] node[below=-1mm] {$b,c$} (q2);
\draw[->,thick] (q3) to[bend left=10] node[below right=-1mm] {$a$} (q4);
\draw[->,thick] (q4) to[bend left=10] node[above  left=-1mm] {$b$} (q3);

\draw[->,thick,dashed] (q0) to[bend left=60] node[above] {$c$} (q1);,
\draw[->,thick,dashed] (q0) ..controls +(0,2) and +(0,2) .. node[below] {$c$} (q3);

\draw[->,thick,dashed] (q1) to[bend left=10] node[above] {$c$} (q2);
\draw[->,thick,dashed] (q1) to[bend left=60] node[below] {$c$} (q0);
\draw[->,thick,dashed] (q1) to[bend left=10] node[above right=-1mm] {$c$} (q4);

\draw[->,thick,dashed] (q2) to[bend left=10] node[below] {$c$} (q1);
\draw[->,thick,dashed] (q2) to[bend left=60] node[above] {$c$} (q3);

\draw[->,thick,dashed] (q4) to[bend left=10] node[below left=-1mm] {$a,c$} (q1);
\draw[->,thick,dashed] (q4) to[bend right=40] node[below right=-1mm] {$a,c$} (q3);

\end{tikzpicture}
\caption{An example minimized \HdTbcBW{} over the alphabet $\Sigma = \{a,b,c\}$}
\label{fig:exampleCoBuchiGFG}
\end{figure}
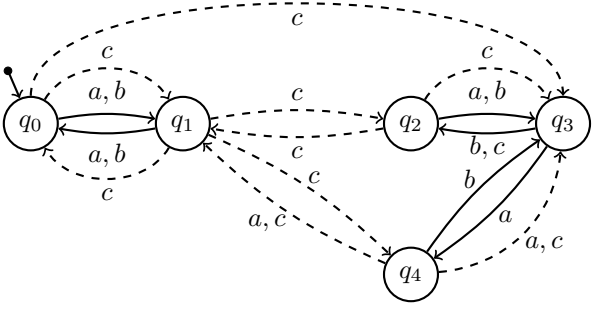

\subsection{Defining rerailing automata}
We are now ready to formally define rerailing automata and show that both minimized~\HdTbcBW{} as well as deterministic parity automata are special cases.

\begin{definition}
\label{def:rerailingautomata}
Let a tuple $\mathcal{A} = (Q,\Sigma,\delta,q_0)$ be a complete automaton structure.
Let furthermore runs, their color sequences and the dominating color along runs be defined for $\mathcal{A}$ as usual.
We say that that $\mathcal{A}$ is a rerailing automaton if
\begin{enumerate}
\item \emph{(Global acceptance:)} The language of $\mathcal{A}$ is defined as $\mathcal{L}_\mathit{rerail}(\mathcal{A}) = \{ w \in \Sigma^\omega \mid \max( \{ \mathsf{dominatingColor}(\rho) \mid (\pi,\rho) \in \mathsf{Runs}(\mathcal{A},w) \}) \text{ is even}  \}$, and 
\item \emph{(Rerailing property:)} For every word $w \in \Sigma^\omega$, every run $\pi = \pi_0 \pi_1 \ldots \in Q^\omega$ with color sequence $\rho \in \NN^\omega$ for $w$, and every index $i$, 
there exists a \emph{rerailing run} $\pi'$ (of $\pi$ and $i$) and a corresponding color sequence $\rho'$ such that $\pi'$ is the same as $\pi$ up to index $i$ and there exists some later index $i' > i$ (\emph{rerailing point}) such that \emph{every} run for $w$ starting with $\pi'$ up to index $i'$ has the same dominating color that is at least as high as the dominating color of $\rho$ and that is even if and only if the word is in $\mathcal{L}_\mathit{rerail}(\mathcal{A})$.
\end{enumerate}
\end{definition}

The second property is the one that gives rerailing automata their name: after every possible prefix run, there is the possibility to extend it in a way that puts it \emph{back on track} such that afterwards, no matter which route a run takes from there onwards, the resulting run stays \emph{on a rail} to have the same dominating color, and this color is even if and only if the word is in the language of the automaton. The dominating colors along runs that have not been put \emph{on track} are ignored, except that they can only be lower than the dominating colors of the rerailing runs derived from them.

We next prove the desired  relationship between rerailing automata and two previously defined automaton types.

\begin{proposition}
Let $\mathcal{A}$ be a deterministic parity automaton. We have that $\mathcal{A}$ is also a rerailing automaton and $\mathcal{L}_\mathit{rerail}(\mathcal{A}) = \mathcal{L}_\mathit{parity}(\mathcal{A})$.
\end{proposition}
\begin{proof}
The rerailing property holds in a trivial way, as every word $w$ induces exactly one run. This is also a rerailing run.
Also, by the definition of $\mathcal{A}_\mathit{parity}$, its dominating color is even if and only if $w$ is in the language represented, so by the definition of which words are accepted by a rerailing automaton, we have $\mathcal{L}_\mathit{rerail}(\mathcal{A}) = \mathcal{L}_\mathit{parity}(\mathcal{A})$.
\end{proof}

\begin{proposition}
\label{proposition:Parity}
Let $\mathcal{A}$ be an automaton structure 
with the properties given in Section~\ref{sec:minimalHdTbcBWProperties} (minimized \HdTbcBW{}). We have that $\mathcal{A}$ is also a rerailing automaton and $\mathcal{L}_\mathit{rerail}(\mathcal{A}) = \mathcal{L}_\mathit{parity}(\mathcal{A})$.
\end{proposition}

\subsection{Some properties of Rerailing Automata}

To lay the foundations for the polynomial-time minimization algorithm for rerailing automata described later in this paper, we state some important properties of rerailing automata next.
In particular, we show that just like minimized~\HdTbcBW{} and deterministic parity automata, rerailing automata track residual languages. %

\begin{lemma}
\label{lem:rerailingAutomatonSuffixLanguageLabeling}
For every rerailing automaton $\mathcal{A} = (Q,\Sigma,\delta,q_0)$ (with all states being reachable), there exists a \emph{residual language labeling} function $f : Q \rightarrow 2^{\Sigma^\omega}$, i.e., such that $f(q_0) = \mathcal{L}_\mathit{rerail}(\mathcal{A})$ and for each $(q,x,q',c) \in \delta$, we have $f(q') = \{ w \in \Sigma^\omega \mid xw \in f(q)\}$.
\end{lemma}
\begin{proof}
Assume the converse, i.e., so that for some states $q,q' \in Q$ both reachable by some finite word $w \in \Sigma^*$, we have that $\mathcal{A}_q$ and $\mathcal{A}_{q'}$ have different languages. Say, w.l.o.g, that for some word $\tilde w \in \Sigma^\omega$, we have $\tilde w \in \mathcal{L}_\mathit{rerail}(\mathcal{A}_q)$ but $\tilde w \notin \mathcal{L}_\mathit{rerail}(\mathcal{A}_{q'})$. 

Since $\tilde w \in \mathcal{L}(\mathcal{A}_{q})$, we have that there exists a run $\pi$ from $q$ for $\tilde w$ so that a rerailing run $\pi'$ for $\pi$ exists such all runs starting with $\pi'$ up to the rerailing point have the same even color. Since there is a run in $\mathcal{A}$ starting with a path from $q_0$ to $q$, followed by $\pi'$, and this run is a rerailing run for $w\tilde w$, we have, by the rerailing property, that $w \tilde w \in \mathcal{L}_\mathit{rerail}(\mathcal{A})$.

Let $w\tilde w \in \mathcal{L}(\mathcal{A})$. Then, we have that all runs for some prefix of $w \tilde w$ can be extended to reach a rerailing point after which the color of the run is always even. 
Since $\tilde w \notin \mathcal{L}(\mathcal{A}_{q'})$, we have that there exists a run $\pi$ from $q'$ for $\tilde w$ with an odd dominating color $c$. Let $\pi' = \pi'_0 \pi'_1 \ldots$ be some rerailing run for $\pi$ with $i'$ as rerailing point, so that all runs starting with $\pi'_0 \ldots \pi'_i$ for $\tilde w$ have the same color. It has to be the same odd color $c$, as $c$ is the largest color along any run from $q'$ for $\tilde w$, and colors of runs after a rerailing point can only be at least as high. But then, by concatenating a run from $q_0$ to $q'$, followed by $\pi'$, we have that from the second rerailing point, all runs have an odd dominating color, which by the rerailing property implies that $w\tilde w \notin \mathcal{L}(\mathcal{A})$, which in turn contradicts the previous point.
\end{proof}

\begin{lemma}
\label{lem:rerailingPoint}
Let $w = w_0 w_1 \ldots$ be a word and $\mathcal{R} = (Q, \Sigma, \delta, q_0)$ be a rerailing automaton. 
For every finite prefix run $\pi_0 \pi_1 \ldots \pi_v$ for $w_0 \ldots w_{v-1}$, there exists a continuation of it, i.e., some infinite run $\pi = \pi_0 \pi_1 \ldots$ for $w$ in $\mathcal{R}$, some SCC $(\hat Q,\hat \delta)$ (\emph{terminal SCC}), and some \emph{terminal rerailing point} $j>v$ such that for every prefix run $\pi' = \pi'_0 \pi'_1 \ldots \pi'_k$ in $\mathcal{R}$ for $w_0 \ldots w_{k-1}$ for some $k>j$ that has $\pi'_j = \pi_j$, we have:
\begin{enumerate}
\item for every transition in $\hat \delta$, there exists an extension of $\pi'$ that takes the transition, and
\item all extensions of $\pi'$ only take transition in $\hat \delta$ (from index $j$ onwards).
\end{enumerate}
\end{lemma}
\begin{proof}
We can find suitable choices of $\pi$ (from index $v$ onwards), $j$, and $(\hat Q,\hat \delta)$ as follows:
First, take an arbitrary run $\pi$ (starting with $\pi_0 \pi_1 \ldots \pi_v$), which by the definition of rerailing automata have to eventually have an index $j$ such that all runs going through $\pi_j$ for $w$ need to have the same dominating color. Then take as first guess for $\hat \delta$ every transition that can be taken along some run for $w$ starting with $\pi_0 \ldots \pi_j$.

Then, repeat the following procedure while maintaining the property that point 2) from the claim holds: Check if $\hat \delta$ already has the needed property.  
In that case, pick as $\hat Q$ all transition ends in $\hat \delta$ and $(\hat Q,\hat \delta)$ is a suitable terminal SCC. Otherwise, there is some extension $\pi_0 \ldots \pi_{k}$ of $\pi_0 \ldots \pi_j$ for $w$ such that after point $k$, only transitions within the current set $\hat \delta$ are taken, and for some transition $t$ in $\hat \delta$, the transition is not taken afterwards by any run extending  $\pi_0 \ldots \pi_{k}$. Then, replace the current value of $j$ by $k$ and remove the transition from $\hat \delta$.

This procedure terminates after a finite number of steps as there are only finitely many transitions in $\hat \delta$ that can be removed.
\end{proof}

Note that a word can have multiple terminal SCCs in a rerailing automaton. 

\section{Decomposing the Structure of Rerailing automata}

Now that rerailing automata have been defined, 
the next step is to show how a minimal such automaton for a target language $L$ can be built.
In this chapter, we show that the key to doing so is a precise characterization of how strongly connected components for the different colors with which words can be recognized need to be \emph{nested} in a rerailing automaton in order to represent $L$.

Starting point is a \emph{chain of co-Büchi automata} (COCOA, \cite{DBLP:conf/tacas/EhlersK24,DBLP:conf/fsttcs/EhlersS22}) representation for a language $L$.
\begin{definition}
\label{def:COCOA}
Let $\mathcal{A}^1, \ldots, \mathcal{A}^n$ (for some $n \in \NN$) be a sequence of co-Büchi automata such that the automata in the chain have a (not necessarily strictly) falling sequence of languages, so that $\mathcal{L}(\mathcal{A}^1) \supseteq \ldots \supseteq \mathcal{L}(\mathcal{A}^n)$. 

Given an infinite word $w$, we say that the chain induces a color of $i$ for $w$ iff either $i=0$ and $w \notin \mathcal{L}(\mathcal{A}^1)$, or $i$ is the greatest index such that $w \in \mathcal{L}(\mathcal{A}^i)$.
If $i$ is even, we say that $w$ is accepted by the chain. The language of the chain is the set of words accepted by the chain.\end{definition}

Chains of co-Büchi automata have recently been proposed as a new representation for arbitrary omega-regular languages \cite{DBLP:conf/fsttcs/EhlersS22}. It was shown that the language of a deterministic parity automaton can be decomposed in a canonical way (by the \emph{natural color of each word}) to a chain of co-Büchi automata in which the chain of languages is strictly falling. All co-Büchi automata in the chain can then be minimized and made canonical with the construction by Abu Radi and Kupferman \cite{DBLP:journals/lmcs/RadiK22} in order to obtain an overall canonical representation of the language. 
For the constructions in the rest of the paper, we will not need the canonicity of a given chain. %

There are multiple reasons for employing such a chain of co-Büchi automata as starting representation in the following:
\begin{enumerate}
\item Obtaining such a chain from a given rerailing automaton (and hence from a deterministic parity automaton as a special case) is relatively simple, as Lemma~\ref{lem:decomposingRerailingAutomatonIntoNonCanonicalCOCOA} below shows. This allows us to define the algorithm for obtaining minimal rerailing automata in the next section in a way that it operates on a chain of co-Büchi automata as starting representation without the construction being unsuitable for minimizing rerailing automata.%
\item We can reuse the main ideas of the minimization procedure for history-deterministic co-Büchi automata with transition-based acceptance by Abu Radi and Kupferman~\cite{DBLP:journals/lmcs/RadiK22}.
While in the first step of processing the COCOA (Section~\ref{subsec:preparation}), we replace the co-Büchi automata with a slightly different automaton type, the minimization process of this alternative automaton type follows the ideas of the minimization procedure by Abu Radi and Kupferman, and the properties of the resulting automata will be used in the correctness argument of the procedure for computing a minimal rerailing automaton in the next section of the paper.
\item A chain of co-Büchi automata essentially represents the colors that some parity automaton assigns to words, but does so in a decomposed way, which we use conceptually in the main result of this section and the main construction of the paper in the next section.
\end{enumerate}
We start by proving the lemma mentioned above.
\begin{lemma}
\label{lem:decomposingRerailingAutomatonIntoNonCanonicalCOCOA}
Let $\mathcal{R} = (Q,\Sigma,\delta,q_0)$ be a rerailing automaton with maximum color $k$ with language $L$.
Let furthermore $C = \{(q,q') \in Q^2 \mid \exists w \in \Sigma^*. \{q,q'\} \subseteq \delta^+(q_0,w) \}$ be a relation that encodes which states can be reached from $q_0$ under the same finite word.
Finally, let for $i \in \{1, \ldots, k\}$ a co-Büchi automaton $\mathcal{A}^i = (Q,\Sigma,\delta^i,q_0)$ be defined with:
\begin{multline*}
\delta^i = \{ (q,x,q',c) \in Q \times \Sigma \times Q \times \NN \mid \exists (q,x,q'',c') \in \delta. \\
(c' < i \wedge (q',q'')  \in C \wedge c=1) \vee (c' \geq i \wedge q'=q'' \wedge c=2) \}
\end{multline*}
Then we have that $\mathcal{A}^1, \ldots, \mathcal{A}^n$ is a chain of history-deterministic co-Büchi automata with language $L$.
\end{lemma}

\begin{proof}
We first of all note that for any $i \in \{1, \ldots, k\}$, all accepting runs of $\mathcal{A}^i$ eventually only take transitions that are also in $\mathcal{R}$ with a color of at most $i$. Also, since $C$ relates equi-reachable states, we have that if and only if a word is accepted by $\mathcal{A}^i$, there is a corresponding run in $\mathcal{R}$ with a dominating color of at least $i$.

To prove that the chain has the right language, let $w$ be a word for which the run with the greatest dominating color  in $\mathcal{R}$ is $c$. Then there is no run for $w$ in $\mathcal{R}$ with dominating color $c+1$. Hence, be the argument above, there cannot be an accepting run in $\mathcal{A}^{c+1}$. Now let us consider a run of color $c$. It exists as accepting run in $\mathcal{A}_c$ as well by the argument above.

It remains to show that for all $i \in \{1, \ldots, k\}$, the automaton $\mathcal{A}^i$ is actually history-deterministic. To prove this, we need that there is a strategy $f$ such that if a word $w$ is accepted, resolving the non-determinism in $\mathcal{A}^i$ by the strategy leads to an accepting run for $w$. We can observe that $\mathcal{A}^i$ is defined such that whenever a rejecting transition is taken, a run can move to \emph{any} state that some other run for the prefix word read so far could be in. 
If $w$ is accepted by $\mathcal{A}^i$, then $w$ has a dominating color of at most $i$ by the reasoning above, and since $\mathcal{R}$ has the property that there is some rerailing run for color $i$, eventually a state along this run is reached along such that \emph{any} extension of it does not take colors smaller than $i$. As a consequence, we can choose a history-determinism resolving strategy that
always takes a transition to some state $q'$ such that the number of steps since last having seen a rejecting transition for some prefix run in $q'$ for the word seen so far is maximized.
If the word is accepted by $\mathcal{A}^i$, with this strategy, eventually a transition to a state from which no more rejecting transitions can occur is taken.
\end{proof}

Note that the construction in Lemma~\ref{lem:decomposingRerailingAutomatonIntoNonCanonicalCOCOA} can be applied in time polynomial in the number of states of $\mathcal{R}$.

\subsection{Preparation: Floating co-Büchi automata}
\label{subsec:preparation}
Starting from the next subsection, we 
relate the structure of a rerailing automaton $\mathcal{R}$ for some language $L$
to the SCCs of automata in some chain of co-Büchi automata for $L$.

In particular, we will reason about in which combination of accepting SCCs runs in different co-Büchi automata can get stuck in for some infinite words. 
In this context, reasoning about through which states runs in automata need to traverse before getting stuck in the accepting SCCs would lead to a larger number of details. %
To keep the exposition simpler, will hence use a proxy model for minimized \HdTbcBW{} that is optimized for such reasoning.

\begin{definition}
Let $L$ be a language and $R^L = (S^L, \Sigma,\delta^L,s_{0})$ be a residual language tracking automaton for $L$. 

We say that some tuple $\mathcal{A} = (Q,\Sigma,\delta,f)$ is a \emph{floating} co-Büchi automaton (for $L$ and $R^L$) if $Q$ is a finite set of states, $\delta : Q \times \Sigma \rightharpoonup Q$ is a partial transition function, and $f : Q \rightarrow S^L$ is a residual language labeling function such that for all $q \in Q, x \in \Sigma$, we have $f(\delta(q,x)) = \delta^L(f(q),x)$. 

We say that some word $w = w_0 w_1 \ldots \in  \Sigma^\omega$ is accepted by $\mathcal{A}$ if there exists some $k \in \NN$ such that there exists some run $\pi = \pi_k \pi_{k+1} \pi_{k+2} \ldots \in Q^\omega$ such that $f(\pi_k) = \delta^L(s^L_0,w_0 \ldots w_{k-1})$ and for all $i \geq k$, we have $\delta(\pi_k,w_k) = \pi_{k+1}$. We also say that $w$ is accepted from position $k \in \NN$ onwards within the state set $\{\pi_k, \pi_{k+1}, \ldots\}$ and within the transition set $\{(\pi_k,w_k),(\pi_{k+1},w_{k+1}), \ldots \}$.

The language of $\mathcal{A}$ (which need not be $L$) is the set of words it accepts.
For every $q \in Q$, we define $\mathsf{Safe}(q)$ to be the set of \emph{safely accepted} words $w \in \Sigma^*$, i.e., those for which $\delta(q,w)$ is defined.
\end{definition}
A floating co-Büchi automaton is, in a sense, a co-Büchi automaton without the rejecting transitions and such that states with different residual language (regarding some language $L$) need to be separate. In this context, $L$ is not necessarily the language of the floating co-Büchi automaton, but could also be some language that the floating co-Büchi automaton helps to define.
A minimized co-Büchi automaton can be translated to an floating co-Büchi automaton as follows:

\begin{definition}
\label{def:residualizedAutomaton}
Let $\mathcal{A}^i = (Q^i,\Sigma,\delta^i,q^i_0)$ be a minimized \HdTbcBW{} of a COCOA for a given language $L^i$, where we denote the set of states being part of an accepting SCC as $\tilde Q^i$.
Let furthermore $R^L = (S^L,\Sigma,\delta^L,s^L_0)$ be a residual language tracking automaton for some language $L$. We define the floating co-Büchi automaton of $\mathcal{A}^i$ as the tuple $\mathcal{A}'^i = (Q'^i,\Sigma,\delta'^i,f)$ with:
\begin{itemize}
\item $Q'^i = \{ (q, s^L) \in Q^i \times S^L \mid \exists w \in \Sigma^*: q \in \delta^i(q^i_0,w), s^L \in \delta^L(s^L_0,w)\}$
\item $\delta'^i = \{((q^i_1,s^L_1),x) \mapsto (q^i_2,s^L_2) \in Q'^i \times \Sigma \times Q'^i \mid (q^i_1,x,q^i_2,2) \in \delta^i, \delta^L(s^L_1,x) = s^L_2 \}$
\item $f((q^i,s^L)) = s^L$ for all $(q^i,s^L) \in Q'^i$
\end{itemize}
\end{definition} 
We can translate a chain of co-Büchi automata to a chain of floating co-Büchi automata for the same language by applying the construction of Def.~\ref{def:residualizedAutomaton} to each automaton, as the following lemma shows.
\begin{lemma}
\label{lem:residualizedAutomaton}

Let a COCOA $\mathcal{A}^1, \ldots, \mathcal{A}^n$ for some language $L$ be given. Let $R^L$ be the residual language tracking automaton for $L$ and
$\mathcal{A}'^1, \ldots, \mathcal{A}'^n$ be floating co-Büchi automata for $\mathcal{A}^1, \ldots, \mathcal{A}^n$ and $R^L$, respectively.

Let the color of a word of $\mathcal{A}'^1, \ldots, \mathcal{A}'^n$ and the language of this chain of floating co-Büchi automata be defined analogously to Def.~\ref{def:COCOA}. 
Then, the languages of the two chains are the same.
\end{lemma}

Note that when translating the chain to floating co-Büchi automata, the color induced by a word can change if the co-Büchi automata distinguish finite prefix words that actually have the same residual language regarding the overall language $L$ being represented with the chain. However, as the lemma shows, this cannot change the language of the chain.

Floating co-Büchi automata have the property that they can be minimized in polynomial time (just as \HdTbcBW{}), as we show next. Later in the paper, when we nest the SCCs of different floating automata, we need to be able to minimize the inner floating co-Büchi automaton in a nesting in a way that it maintains some labeling of inner automaton states with outer automaton states. To support this, the following minimization lemma takes such a labeling into account. This labeling will be represented by a marking in addition to the floating co-Büchi automaton tracking some residual language.

\begin{definition}
A tracking automaton $R = (Q^R,\Sigma,\delta^R,q^R_0)$ is a finite-state automaton with the set of states $Q^R$, the transition function $\delta : Q \times \Sigma \rightarrow Q^R$, and some initial state $q^R_0$. Runs of the tracking automaton are defined as usual.
\end{definition}

\begin{lemma}
\label{lem:minimizationLemma}
Let a $\mathcal{A} = (Q,\Sigma,\delta,f)$ be a floating Büchi automaton for some residual language tracking automaton $R^L$.
Let furthermore some other tracking automaton $R = (Q^R,\Sigma,\delta^L,q^R_0)$ be given and $m : Q \rightarrow Q^R$ be some marking function of the states of $\mathcal{A}$ such that for each $q \in Q$ and $x \in \Sigma$, we have $\delta^R(m(q),x) = m(\delta(q,x))$.

We can compute, in polynomial time, an automaton $\mathcal{A}' = (Q',\Sigma,\delta',f')$ and some marking function $m' : Q' \rightarrow Q^R$ that minimizes $\mathcal{A}$ modulo $m$, i.e., that is as small as possible while ensuring that the languages of $\mathcal{A}$ and $\mathcal{A}'$ are the same, and whenever $\mathcal{A}'$ has a run $\pi'_k \pi'_{k+1} \ldots \in Q'^\omega$ for some word $w = w_0 w_1 \ldots$, then there exists a run $\pi_k \pi_{k+1} \in Q^\omega$ in $\mathcal{A}$ such that $m(\pi_j) = m'(\pi'_j)$ for all $j \geq k$.
\end{lemma}

For analyzing the structure of a minimal rerailing automaton representing the language of a chain of floating co-Büchi automata, we will need to find parts in a rerailing automaton that correspond to the SCCs in the floating automata (and some other floating automata that we derive from the chain's automata). Terminal SCCs (Lemma~\ref{lem:rerailingPoint}) can fulfill this role if we can define a word that characterizes the SCC, for which we employ the following definition/lemma pair:

\begin{definition}
\label{def:SaturatingWord}
Let $\mathcal{A} = (Q,\Sigma,\delta,f)$ be a floating co-Büchi automaton for some RLTA $R^L = (S^L,\Sigma,\delta^L,s^L_0)$ and $(Q',\delta')$ be a maximal SCC in $\mathcal{A}$.

We say that some word $w = w_0 w_1 \ldots \in \Sigma^\omega$ is a \emph{saturating word} together with a \emph{saturating run} $\pi = \pi_k \pi_{k+1} \ldots$ and \emph{saturation point} $k \in \NN$ for $(Q',\delta')$ if 
\begin{enumerate}
\item[(a)] $\pi$ is a run of $\mathcal{A}$ for $w_k w_{k+1} \ldots$ and $f(\pi_k) = \delta^L(s^L_0,w_0 \ldots w_{k-1})$,
\item[(b)] for every sequence of transitions $q_1 \xrightarrow{x_1} \ldots \xrightarrow{x_{n-1}} q_n$ in $\delta'$, (i.e., we have $(q_{i},x,q_{i+1}) \in \delta'$ for all $1 \leq i < n$), there exist infinitely many $j>k$ such that $\pi_{j} w_j \ldots w_{j+n-1} \pi_{j+n} = q_1 x_1 \ldots x_{n-1} q_n$.
\end{enumerate}
\end{definition}
Informally, a saturating word for a maximal SCC of a floating co-Büchi automaton is one that takes all possible finite sequences of transitions within a maximal SCC infinitely often. 

\begin{lemma}
Let $\mathcal{A} = (Q,\Sigma,\delta,f)$ be a floating co-Büchi automaton for some RLTA $R^L = (S^L,\Sigma,\delta^L,S^L_0)$ and $(Q',\delta')$ be a maximal SCC in $\mathcal{A}$. There exists a saturating word and a corresponding run for $\mathcal{A}$ and $(Q',\delta')$.
\end{lemma}
\begin{proof}
We construct $w$ (the word) and $\pi$ (the run as follows):
\begin{itemize} 
\item We let $w$ start with a prefix $w_0 \ldots w_{k-1}$ such that for some state $q \in Q'$, we have $f(q) = \delta^L(s^L_0,w_0 \ldots w_{k-1})$ and start $\pi$ with $\pi_k = q$.
\item We enumerate all sequences of transitions $q_1 \xrightarrow{x_1} \ldots \xrightarrow{x_{n-1}} q_n$ ordered by length (shortest first). For each such sequence, we add a finite path positioning the (current) end of $\pi$ to $q_1$ and add a corresponding finite word to $w$. Then, we add the sequence elements to $\pi$ and $w$.
\end{itemize}
Every possible finite sequence of transitions appears infinitely often the resulting word/run pair as since $(Q',\delta')$ is an SCC, the sequence is the prefix of infinitely many sequences appearing later in the combination. 
\end{proof}

\color{black}

Floating co-Büchi automata have two more properties that we will need in the following:

\begin{lemma}
\label{lem:synchronizingWord}
Let $\mathcal{A} = (Q,\Sigma,\delta,f)$ be a minimized floating co-Büchi automaton for some $R^L$ and $q$ and $q'$ be two states with $f(q)=f(q')$ and $\mathsf{Safe}(q) \subset \mathsf{Safe}(q')$. Then, there exists some finite word $w$ such that $\delta(q,w) = \delta(q',w)$ (\emph{synchronizing word}).
\end{lemma}

\begin{lemma}
\label{lem:uniquenessLemma}
Let $\mathcal{A} = (Q,\Sigma,\delta,f)$ be a minimized floating co-Büchi automaton for some $R^L$, $m$ be some marking function, and $Q' \subseteq Q$ be the states of a maximal strongly connected component in $\mathcal{A}$. There exists a state $q$ in $Q'$ such that for some word $w = w_0 w_1 \ldots \in \Sigma^\omega$, $q$ is the only state in $\mathcal{A}$ with marking $m(q)$ and the same $R^L$ label as $q$ from which $w$ is safely accepted.
\end{lemma}

\subsection{How SCCs need to be nested in rerailing automata}

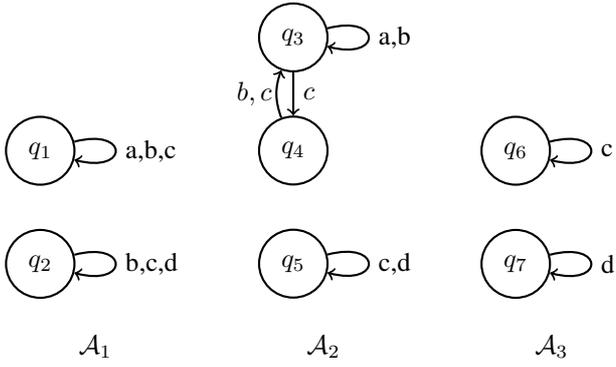
\begin{figure}[t]
\centering
\begin{tikzpicture}

\node[anchor=south] (a1) at (0,0) {
  \begin{tikzpicture}
  \draw (0,0) node[draw,shape=circle,minimum size=0.9cm,thick] (q1) {$q_1$};
  \draw[->,thick] (q1) to[loop right] node[right] {a,b,c} (q1) ;
  
  \draw (0,-1.5) node[draw,shape=circle,minimum size=0.9cm,thick] (q2) {$q_2$};
  \draw[->,thick] (q2) to[loop right] node[right] {b,c,d} (q2) ;

  \end{tikzpicture}
};

\draw (0,-0.5) node {$\mathcal{A}_1$};

\node[anchor=south] (a2) at (3.0,0) {
  \begin{tikzpicture}
  \draw (0,0) node[draw,shape=circle,minimum size=0.9cm,thick] (q3) {$q_3$};
  \draw[->,thick] (q3) to[loop right] node[right] {a,b} (q3) ;
  
  \draw (0,-1.5) node[draw,shape=circle,minimum size=0.9cm,thick] (q4) {$q_4$};
  
  \draw[->,thick] (q3) to[bend left=0] node[right] {$c$} (q4);
  \draw[->,thick] (q2) to[bend left=20] node[left=-2pt] {$b,c$} (q1);

  \draw (0.0,-3.0) node[draw,shape=circle,minimum size=0.9cm,thick] (q5) {$q_5$};
  \draw[->,thick] (q5) to[loop right] node[right] {c,d} (q5) ;

  \end{tikzpicture}
};

\draw (3.0,-0.5) node {$\mathcal{A}_2$};

\node[anchor=south] (a3) at (6,0) {
  \begin{tikzpicture}
  \draw (0,0) node[draw,shape=circle,minimum size=0.9cm,thick] (q1) {$q_6$};
  \draw[->,thick] (q1) to[loop right] node[right] {c} (q1) ;
  
  \draw (0,-1.5) node[draw,shape=circle,minimum size=0.9cm,thick] (q2) {$q_7$};
  \draw[->,thick] (q2) to[loop right] node[right] {d} (q2) ;

  \end{tikzpicture}
};

\draw (6,-0.5) node {$\mathcal{A}_3$};

\node (k) at (13,1) {

};

\end{tikzpicture}
\caption{An example chain of minimized floating co-Büchi automata for a uniform residual lanaguage over $\Sigma= \{a,b,c,d\}$}
\label{fig:chain}
\end{figure}

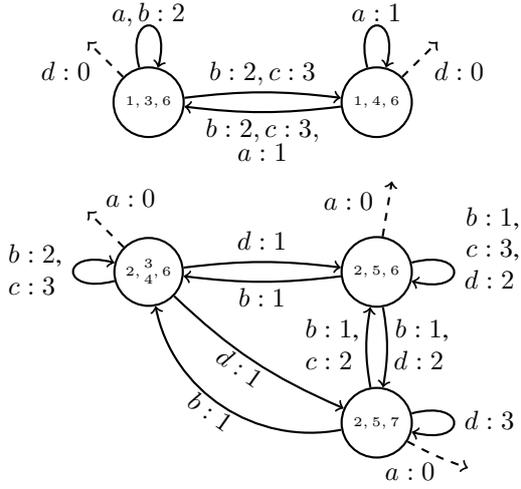
\begin{figure}[t]
\centering
\begin{tikzpicture}

  \draw (0,0.25) node[draw,shape=circle,minimum size=0.9cm,thick] (q1) {\tiny $1,3,6$};
  \draw[->,thick] (q1) to[loop above] node[above=-2mm] {\begin{tabular}{c} $a,b:2$\end{tabular}} (q1) ;
  \draw (3,0.25) node[draw,shape=circle,minimum size=0.9cm,thick] (q2) {\tiny $1,4,6$};

  \draw[->,thick] (q1) to[bend left=7] node[above=-1mm] {\begin{tabular}{c} $b:2, c:3$\end{tabular}} (q2) ;

  \draw[->,thick] (q2) to[bend left=7] node[below=-1mm] {\begin{tabular}{c} $b:2, c:3, \vspace{-1mm} $\\ $a:1$\end{tabular}} (q1) ;

  \draw[->,thick] (q2) to[loop above] node[above=-2mm] {\begin{tabular}{c} $a:1$\end{tabular}} (q2) ;

  \draw[->,thick,dashed] (q1) to node[below left=-1.5mm] {\begin{tabular}{c} $d:0$\end{tabular}} ($(q1)+(-0.8,0.8)$);

  \draw[->,thick,dashed] (q2) to node[below right=-1.5mm] {\begin{tabular}{c} $d:0$\end{tabular}} ($(q2)+(+0.8,0.8)$);

  \draw (0,-2) node[draw,shape=circle,minimum size=0.9cm,inner sep=0pt,thick] (q3) {\tiny $2,\begin{matrix}{3}\\{4}\end{matrix},6$};

  \draw (3,-2) node[draw,shape=circle,minimum size=0.9cm,thick] (q4) {\tiny $2,5,6$};
  
  \draw (3,-4) node[draw,shape=circle,minimum size=0.9cm,thick] (q5) {\tiny $2,5,7$};

  \draw[->,thick] (q3) to[loop left] node[left=-2mm] {\begin{tabular}{c} $b:2$,\\ $c:3${\color{white},}\end{tabular}} (q3) ;
  
  \draw[->,thick] (q4) to[loop right] node[right=-2mm,yshift=8pt] {\begin{tabular}{c} $b:1$,\\ $c:3$, \\ $d:2${\color{white},} \end{tabular}} (q4) ;

  \draw[->,thick] (q5) to[loop right] node[right=-2mm] {\begin{tabular}{c} $d:3$ \end{tabular}} (q5) ;

  \draw[->,thick] (q3) to[bend left=7] node[above=-1mm] {\begin{tabular}{c} $d:1$\end{tabular}} (q4) ;
  \draw[->,thick] (q4) to[bend left=7] node[below=-1mm] {\begin{tabular}{c} $b:1$\end{tabular}} (q3) ;

  \draw[->,thick] (q4) to[bend left=10] node[right=-2.5mm] {\begin{tabular}{c} $b:1$, \\ $d:2${\color{white},}\end{tabular}} (q5) ;
  \draw[->,thick] (q5) to[bend left=10] node[left=-2.5mm] {\begin{tabular}{c} $b:1$, \\ $c:2${\color{white},}\end{tabular}} (q4) ;

  \draw[->,thick] (q3) to[bend left=-10] node[below left=-5mm] {\rotatebox{-35}{\begin{tabular}{c} $d:1$\end{tabular}}} (q5) ;
  \draw[->,thick] (q5) to[bend left=45] node[below left=-5mm] {\rotatebox{-35}{\begin{tabular}{c} $b:1$\end{tabular}}} (q3) ;

  \draw[->,thick,dashed] (q3) to node[above right=-1.5mm,pos=0.9] {\begin{tabular}{c} $a:0$\end{tabular}} ($(q3)+(-0.8,0.8)$);

  \draw[->,thick,dashed] (q4) to node[left=-1.5mm,pos=0.6] {\begin{tabular}{c} $a:0$\end{tabular}} ($(q4)+(0.2,1.2)$);

  \draw[->,thick,dashed] (q5) to[bend right=5] node[inner sep=0pt,below left =-0.5mm,pos=0.7] {\begin{tabular}{c} $a:0$\end{tabular}} ($(q5)+(1.2,-.6)$);

\end{tikzpicture}

\caption{A minimal rerailing automaton over the alphabet $\Sigma = \{a,b,c\}$. Dashed lines represent cases in which transitions to all states (including the state from which the transition starts) exist. }
\label{fig:rerailing}
\end{figure}

We discuss the connection between a chain of floating co-Büchi automata for a language and a corresponding rerailing automaton by means of an example shown in Figure~\ref{fig:chain} and Figure~\ref{fig:rerailing}. To allow focusing on how SCCs need to be nested,
we use a language $L$ with only a single residual language, so that in Figure~\ref{fig:rerailing}, the residual languages of 
the states do not need to be marked.
Figure~\ref{fig:chain} contains a chain of floating co-Büchi automata $\mathcal{A}^1, \mathcal{A}^2, \mathcal{A}^3$ with $\mathcal{L}(\mathcal{A}^1) \supset \mathcal{L}(\mathcal{A}^2) \supset \mathcal{L}(\mathcal{A}^3)$ for this language $L$, and a minimal rerailing automaton for $L$ is shown in Figure~\ref{fig:rerailing}. 

Let us consider the words that contain both the letters $a$ and $d$ infinitely often. Such words are rejected by $\mathcal{A}^1$, and hence by the definition of acceptance of such a chain in Lemma~\ref{def:residualizedAutomaton}, the rerailing automaton $\mathcal{R}$ for $L$ has to accept all such words.
When looking at $\mathcal{R}$ in Figure~\ref{fig:rerailing}, it can be seen  that this is the case. From the two states at the top, reading a $d$ leads to taking a transition with color $0$, and such a transition is the only way to leave the states at the top. Likewise, from the states at the bottom, reading an $a$ will lead to taking a transition with color $0$, and that is the only way to leave these states. So every run for a word with infinitely many $a$s and $d$s has a dominating color of $0$.

With the $20$ transitions with color $0$ not explicitly shown in Figure~\ref{fig:rerailing}, it is easy to see that the automaton is \emph{split} into a part corresponding to state $q_1$ in $\mathcal{A}^1$ and a part corresponding to state $q_2$. Transitions between these parts have an even color that is smaller than all colors within the two parts.
Within the top part, the rerailing automaton tracks the states of $\mathcal{A}^2$, except that there is no state for the $q_1$/$q_5$ combination as whenever for a word, there are accepting runs from both of them, they are also accepted with the SCC of $\mathcal{A}^2$ comprising $q_3$ and $q_4$, so the rerailing automaton does not need a state for the $q_1$/$q_5$ combination. In the lower half of Figure~\ref{fig:rerailing}, however, there are states distinguishing the SCCs of $\mathcal{A}^2$, but here, the rerailing automaton states $q_3$ and $q_4$ can be the same because the two states' safe languages only differ by words that are not safely accepted by $q_2$. 

We show with the \emph{central nesting theorem} below that rerailing automata track in which SCCs of a chain of floating co-Büchi automata runs can get stuck, but the floating co-Büchi automata can be minimized by \emph{context}, just as in the example, in the $\{q_2\}$ context, the states $q_3$ and $q_4$ in $\mathcal{A}^2$ are merged.

As a first step, we characterize a property of strongly connected components in rerailing automata that are terminal for some saturating word of a co-Büchi automaton.

\begin{lemma}
\label{lem:MLabeling}
Let $\mathcal{R} = (Q,\Sigma,\delta,q_0)$ be a rerailing automaton and $\mathcal{A}^1 = (Q^1,\Sigma,\delta^1,f^1_0)$ be a floating co-Büchi automaton. 

Let furthermore $(Q'^1,\delta'^1)$ be a maximal accepting SCC in $\mathcal{A}^1$. Let $w$ be a word that such that all finite \hyperref[def:consecutiveSubword]{consecutive subwords} of a saturating word for $(Q'^1,\delta'^1)$ \hyperref[def:appears]{appear} in it at infinitely many positions, and $(\hat Q,\hat \delta)$ be a terminal SCC of $\mathcal{R}$ for $w$ for a run $\pi$ with terminal rerailing point $k \in \NN$ (with the properties given in Lemma~\ref{lem:rerailingPoint} and a prefix run of $\epsilon$). Let, without loss of generality, the start of $\mathcal{A}^1$'s run for $w$ be before position $k$.

Then there exists some mapping $m: Q'^1 \rightarrow 2^{\hat Q}$ such that
\begin{enumerate}
\item for every state $q^1 \in Q'^1$, $m(q^1)$ is non-empty
\item for every state $q^1 \in Q'^1$, every $q \in m(q^1)$ and $x_0 \ldots x_n \in \Sigma^*$, if $\delta'^1(q^1,x_0 \ldots x_n)$ is defined, then for every 
sequence $q'_0 \xrightarrow{x_0} q'_1 \ldots q'_{n} \xrightarrow{x_n} q'_{n+1}$ with 
$q'_0 = q$ and $(q'_i,x_i,q'_{i+1},c_i) \in \delta$ for all $0 \leq i \leq n$ for some $c_i \in \NN$, we have that all such transitions $(q'_i,x_i,q'_{i+1},c_i)$ are in $\hat \delta$.
\end{enumerate}
\end{lemma}

The main theorem below shows that rerailing automata have SCCs that correspond 1-to-1 to some maximal SCCs in floating co-Büchi automata. The following definition makes precise what this means:

\begin{definition}[Characterizing SCC]
\label{def:characterizingSCC}
Let $\mathcal{R} = (Q,\Sigma,\delta,q_0)$ be a rerailing automaton for some language $L$ and $\mathcal{A} = (Q^\mathcal{A},\Sigma,\delta^\mathcal{A},f^A)$ be a floating co-Büchi automaton for some residual language tracking automaton $R^L$ of $L$.

Let furthermore $(Q', \delta')$ be an SCC of $\mathcal{R}$ and $(Q'^{\mathcal{A}},\delta'^{\mathcal{A}})$ be a maximal accepting SCC of $\mathcal{A}$. We say that $(Q', \delta')$ \emph{characterizes} $(Q'^{\mathcal{A}},\delta'^{\mathcal{A}})$ if there exists some \emph{characterization function} $f : Q' \rightarrow Q'^{\mathcal{A}}$ such that for each $q \in Q'$ and every character $x \in \Sigma$, we have either have:
\begin{itemize}
\item there exists a transition from $f(q)$ for x in $\delta'^{\mathcal{A}}$, there exists a transition for $x$ from $q$ inside $\delta'$, and for all $q'' \in Q$ such that $(q,x,q'',c) \in \delta$ for some $c \in \NN$, we have that $(q,x,q'',c) \in \delta'$, or
\item there does not exist a transition from $f(q)$ for $x$ in $\delta'^{\mathcal{A}}$, and 
 for all transitions $(q,x,q'',c) \in \delta$ for some $c \in \NN$ and $q'' \in Q$, we have $(q,x,q'',c) \notin \delta'$.
\end{itemize}
Furthermore, for every $q' \in Q'$, we have that the residual language of the state is $f^A(f(q'))$.
\end{definition}

We need one more definition in order to be able to make precise what it means to nest maximal SCCs in different floating co-Büchi automata:
\begin{definition}
\label{def:ALabeledFloatingAutomaton}
Let $\mathcal{A}^A = (Q^A,\Sigma,\delta^A,f^A)$ and $\mathcal{A}^B = (Q^B,\Sigma,\delta^B,f^B)$ be two floating co-Büchi automata for some residual language tracking automaton $R^L$.
We say that $\mathcal{A}^B$ is $\mathcal{A}^A$-labeled with some labeling function $f : Q^B \rightarrow Q^A$ if for each $q \in Q^B$ and $x \in \Sigma$, if $\delta^B(q,x)=q'$ for some $q'$, then we have $\delta^A(f(q),x)=f(q')$.
\end{definition}

Note that if $\mathcal{A}^B$ is $\mathcal{A}^A$-labeled, then all words accepted by $\mathcal{A}^B$ are also accepted by $\mathcal{A}^A$.
By Lemma~\ref{lem:minimizationLemma}, we can minimize such an automaton $\mathcal{A}^B$ modulo some $\mathcal{A}^A$-labeling.
We are now prepared for the main main theorem:

\begin{theorem}[Central nesting theorem]
\label{thm:centralNestingThm}
Let $\mathcal{R} = (Q,\Sigma,\delta,q_0)$ be a rerailing automaton for some language $L$, $R^L$ be a residual language tracking automaton for $L$, $\mathcal{A}^A = (Q^A,\Sigma,\delta^A,f^A)$ be a floating co-Büchi automaton with a single maximal SCC, $w$ be a saturating word for $(Q^A, \delta^A)$, and $(Q', \delta')$ be an SCC of $\mathcal{R}$ for $w$ that characterizes $(Q^A, \delta^A)$ with a characterization function $f : Q' \rightarrow Q^A$.

Let furthermore $\mathcal{A}^B = (Q^B,\Sigma,\delta^B,f^B)$ be an $\mathcal{A}^A$-labeled minimized floating co-Büchi automaton (with labeling function $f':Q^B \rightarrow Q^A$) 
such that saturating words for $\mathcal{A}^A$ are uniform w.r.t containment in $L$, words rejected by $\mathcal{A}^B$ but accepted by $\mathcal{A}^A$ are uniform w.r.t.~containment in $L$, and containment in $L$ of these two groups of words differs.

Then, each of the maximal SCCs $(Q'^B_1, \delta'^B_1), \ldots, (Q'^B_m, \delta'^B_m)$ of $\mathcal{A}^B$ have disjoint characterizing SCCs of $\mathcal{R}$ within $(Q', \delta')$.

\end{theorem}

\begin{proof}
For all $1 \leq i \leq m$, let an SCC $(Q'^B_i, \delta'^B_i)$ of $\mathcal{A}^B$ be given, $q^B \in Q'^B_i$ be some state and $q \in Q'$ be given with $f'(q) = f(q^B)$. Let furthermore $w^i = w^i_0 w^i_1 \ldots \in \Sigma^\omega$ be a saturating word for $(Q'^B_i,\delta'^B_i)$ be given that has some run $\pi = \pi_k \pi_{k+1} \ldots$ with $\pi_k = q^B$.

Let us consider a terminal SCC $(\hat Q^i, \hat \delta^i)$ of $\mathcal{R}$ for $w^i_k w^i_{k+1} \ldots$ when using  $q$ as initial state. Among all possible terminal SCCs for the word, we pick one with the highest possible (among the SCCs) lowest color occurring along the transitions of the SCC.
By the fact that $\mathcal{A}^B$ is $\mathcal{A}^A$-labeled and $(Q',\delta')$ characterizes the only SCC in $\mathcal{A}^A$, we know that no run for $w^i_k w^i_{k+1} \ldots \ldots$ from $q$ leaves $(Q',\delta')$, and hence $(\hat Q^i, \hat \delta^i)$ is contained in $(Q',\delta')$.

Let furthermore $r^i \subseteq \hat Q^i \times Q'^B_i$ be a relation containing exactly those state pairs $(q',q'^B)$ such that along some run $\pi^i = \pi^i_0 \ldots$ from $q$ for $w^i_k w^i_{k+1} \ldots$, we have infinitely many $j \in \NN$ such that $\delta^B(q^B,w^i_k \ldots w^i_{k+j}) = q'^B$ and $\pi^i_j = q'$. Let the color of the run be $c^i$, which by the definition of terminal SCCs and rerailing automata is then also the minimal color among the transitions in $\hat \delta^i$. Without loss of generality, we assume that $\pi^i$ is one of the runs taking all transitions in $\hat \delta^i$ infinitely often. Note that by Def.~\ref{def:characterizingSCC} and Def.~\ref{def:ALabeledFloatingAutomaton}, for all $(q',q'^B) \in r^i$, we have $f'(q') = f(q'^B)$.

Let in the following $r = r^1 \cup \ldots \cup r^m$.

\textbf{Analyzing $r$:} We now prove that $r$ satisfies a non-overlap property in order to show a bit later that we can use $r$ to find the characterizing SCCs in $\mathcal{R}$ for the maximal SCCs of $\mathcal{A}^B$.

Assume that for some $1 \leq i \leq m$ and $1 \leq l \leq m$, there exist $q' \in \hat Q^i$ and $q'^B \in Q^B_i, q''^B \in Q^B_l$ with $q'^B \neq q''^B$ and $\{(q',q'^B),(q',q''^B)\} \subseteq r$.
Let, without loss of generality, $q''^B$ be a state such that the safe language of $q''^B$ is not a subset of the safe language of $q'^B$ (they cannot have the same safe language as $\mathcal{A}^B$ is minimized modulo some $\mathcal{A}^A$-labeling and the residual languages, and $q'^B$ and $q''^B$ have the same $\mathcal{A}^A$-labeling and the same residual language since they are related to the same state $q'$), so there is some finite word $u$ with $\delta^B(q''^B,u)$ being defined but $\delta^B(q'^B,u)$ is not.

We mix the words $w^i$ and $w^l$ into a new word $w^{i \oplus l}$ and a new run $\pi^{i \oplus l}$ in $\mathcal{R}$ such that the word is not accepted by $\mathcal{A}^B$ but the run takes all transitions of the SCCs $(\hat Q_i, \hat \delta_i)$ and $(\hat Q_l, \hat \delta_l)$ infinitely often.
Let us consider an extension of $u$, namely $u^+$ such that there is a run (part) from $q'$ to $q'$ in $(\hat Q^i,\hat \delta^i)$ and let $\pi^u$ be the corresponding run part.

We split $w^i$/$\pi^i$ and $w^l$/$\pi^l$ into pieces such that in piece number $j \in \NN$, all finite words of length $j$ that appear infinitely often in $w^i$ and $w^l$, respectively, are contained at least once. Since both words are saturating words, such a split can be found.

We then mix $w^i$ and $w^l$ and $\pi^i$ and $\pi^l$ by taking, in a round-robin fashion, first one part from the word in $(\hat Q_i, \hat \delta_i)$, one part from the word in $(\hat Q_l, \hat \delta_l)$, and then $u^+$/$\pi^+$.

The resulting word/run combination is such that the run stays in $(Q',\delta')$, and by the fact that $(Q',\delta')$ is characterizing for $\mathcal{A}^A$, is accepted by $\mathcal{A}^A$. The word is also rejected by $\mathcal{A}^B$. To see this, we can consider two cases:

Either $q''^B$ and $q'^B$ are in different SCCs of $\mathcal{A}^B$ (so $i \neq l$), and then the rerailing words contain, infinitely often, sub-words that are only safely accepted from one state in $\mathcal{A}^B$. This is because $\mathcal{A}^B$ is minimized, so that each SCC contains a state with a unique safely accepted word (Lemma~\ref{lem:uniquenessLemma}). Since such words appears infinitely often in $w^{i \oplus l}$ for both maximal SCCs of $\mathcal{A}^B$, we have that the word is rejected by $\mathcal{A}^B$.

Alternatively, we have $i=l$, and then the word $w^{i \oplus l}$ contains infinitely often sequences of synchronizing words (Lemma~\ref{lem:synchronizingWord}) and words that rule out some states such that any run from some state in the SCC $(\hat Q^l, \hat \delta^l)$ before this sequence is in a single known state afterwards. When transitioning between the word parts, any run in $\mathcal{A}^B$ is then in the state $q'^B$ when $u$ is the next part of $w^{i \oplus l}$. This happening infinitely often also precludes $w^{i \oplus l}$ from being accepted by $\mathcal{A}^B$.

By acceptance uniformity of words with the criteria in the claim, we have that either (a) $c^i$ is not the right color for a saturating word for the SCC $(Q'^B_i,\delta'^B_i)$, but then such a word having the terminal SCC $(\hat Q^i, \hat \delta^i)$ contradicts that $\mathcal{R}$ has the correct language (or that $\mathcal{R}$ has the rerailing property), (b) $c^l$ is not the right color for a saturating word for the SCC $(Q'^B_l,\delta'^B_l)$, but then such a word having the terminal SCC $(\hat Q^l, \hat \delta^l)$ contradicts that $\mathcal{R}$ has the correct language (or that $\mathcal{R}$ has the rerailing property), or (c) there is a different terminal SCC for which $w^{i \oplus l}$ has a different color than $\min(c^i,c^l)$. 

By the fact that a run for $w^{i \oplus l}$ in exists in $(Q',\delta')$ with color $\min(c^i,c^l)$, by the fact that runs for $w^{i \oplus l}$ are stuck in $(Q',\delta')$ (as $(Q',\delta')$ characterizes the SCC of $\mathcal{A}^A$, and $w^{i \oplus l}$  is accepted by $\mathcal{A}^A$), and by the rerailing property, there has to exist some terminal SCC $(\hat Q'', \hat \delta'')$ within $(Q', \delta')$ (i.e., such that $\hat \delta'' \subseteq \delta'$) with which $w^{i \oplus l}$ is recognized, and the greatest color along the transitions in $\hat \delta''$ is a color $c' > \min(c^i,c^l)$.

We can show that in this case, either a $w^i$ has a run with a higher color than $c^i$, which contradicts the assumption from earlier that we picked an accepting terminal SCC for $w^i$ with the highest possible lowest color, or $w^l$ has a run with a higher color than $c^l$ for the same reason. 

Let, w.l.o.g., assume that $c'>c^i$. Since $w^{i \oplus l}$ satisfies the conditions from Lemma~\ref{lem:MLabeling} with the SCC $(\hat Q'', \hat \delta'')$ of $\mathcal{R}$, there exists a labeling function $m$ mapping each state in $Q'^B_i$ to a state in $(\hat Q'', \hat \delta'')$ such that \emph{every} run from there stays in the SCC. Since each such run can only take transitions in $\delta''$, the smallest color of these is greater than $c^i$, which completes the contradiction. The other case (for $w^l$/$c^l$ is analogous).

\textbf{Using $r$:} 
With the property of $r$ in place that every state $q \in Q'$ can only be labeled by one state of $\mathcal{A}^B$, we can now partition the states of $Q'$ into those for each accepting SCC of $\mathcal{A}^B$ (and perhaps some that are not labeled by any state in $\mathcal{A}^B$). 
Let $Q'^i$ be these states for an accepting SCC $(Q^B_i, \delta^B_i)$ of $\mathcal{A}^B$. By the fact that $(\hat Q'_i, \hat \delta'_i)$ is a terminal SCC of $w^i$, we have that whenever for some $q,q^B$, we have $(q,q^B) \in r$, then for all $q'$ with $(q,x,q',c) \in \delta$ for some $x$ and $c$, we have that $(q',q'^B) \in r$ for the only $q'^B$ with $\delta^B(q^B,x) = q'^B$ (if this transition exists). This means that $Q'^i$ together with the transitions of $\hat \delta$ between them for transitions corresponding to $\delta^B_i$ is an SCC that characterizes $(Q^B_i, \delta^B_i)$, and they are disjoint by the property of $r$ that we found above.
\end{proof}

\section{Building minimal rerailing automaton}

With the central rerailing theorem from the previous section in place, we are now ready to define %
a procedure for computing a minimal rerailing automaton from a chain of co-Büchi automata.
Its core property is that it only separates those states that have to be separate according to the central nesting theorem. %

\subsection{Preparation: Computing the residual language tracking automaton for a given chain of co-Büchi automata}

Let a chain of co-Büchi automata (COCOA) $\mathcal{A}^1, \ldots, \mathcal{A}^n$ for a language $L$ represented by the chain according to Def.~\ref{def:COCOA} be given. This COCOA could, for instance, have been obtained by decomposing an existing rerailing automaton (or deterministic parity automaton) in polynomial time (Lemma~\ref{lem:decomposingRerailingAutomatonIntoNonCanonicalCOCOA}). Let each of the automata in this chain have already been minimized with the minimization algorithm by Abu Radi and Kupferman~\cite{DBLP:journals/lmcs/RadiK22}, so they have the properties stated in Section~\ref{sec:minimalHdTbcBWProperties}.

We first translate all automata in the chain to floating co-Büchi automata, without changing the language $L$ of the chain (Lemma~\ref{lem:residualizedAutomaton}).
Definition~\ref{def:residualizedAutomaton}, which states how to do so, requires a residual language tracking automaton to be given in addition to the history-deterministic co-Büchi automaton. We employ the residual language tracking automaton for $L$ for this purpose rather than residual language tracking automata for the input co-Büchi automata.

In order to do so, we however first need to obtain a residual language tracking automaton $R^L$ for $L$. A naive approach to obtaining it is to 
compute residual language tracking automata for every history-deterministic co-Büchi automaton in the chain and then take the product of these residual language tracking automata. 
This naive approach has two problems: (a) it can happen that
two distinct states in the product represent the same residual language with respect to $L$, and (b) the product built can be of size exponential in the sizes of the input automata despite the fact that $R^L$ may be of size only polynomial in the sum of the input automaton sizes.
While problem (a) can be addressed by minimizing the product after building it, problem (b) means computing $R^L$ may take time exponential in the input size even in cases in which such a long computation time is not actually needed. 

Figure~\ref{fig:exampleSuffixLanguageBlowup} shows an example 
chain of (minimized) co-Büchi automata with two automata that demonstrates that $R^L$ can be smaller that the product of the residual language tracking automata of the individual co-Büchi automata in the chain. While the language represented by this chain is trivial (it is the universal language over $\Sigma = \{a,b\}$), the example highlights the problem that the relatively complex algorithm below (in Figure~\ref{fig:residualComputationAlgorithm}) for computing $R^L$ solves. 
The overall language represented by the chain in Figure~\ref{fig:exampleSuffixLanguageBlowup} is the universal language, but not every word is accepted at the same level of the chain. Words starting with $a$ are accepted by both automata in the chain, while words starting with a $b$ are accepted by none of them. In both cases, Def.~\ref{def:COCOA} states that the respective words are in the represented language.

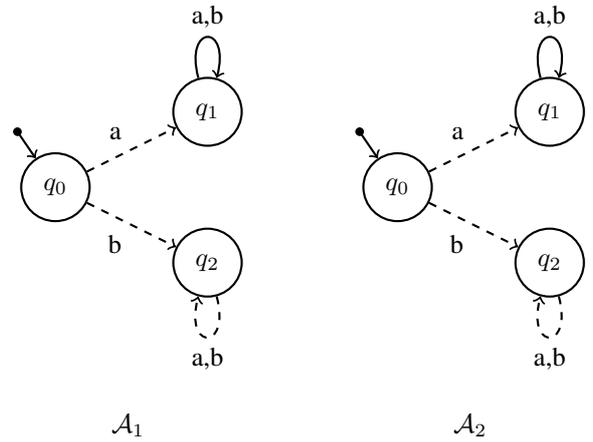
\begin{figure}[t]
\centering
\begin{tikzpicture}

\node[anchor=south] (a1) at (0,0) {
  \begin{tikzpicture}
  \draw (0,0) node[draw,shape=circle,minimum size=0.9cm,thick] (q0) {$q_0$};
  \draw (2,1) node[draw,shape=circle,minimum size=0.9cm,thick] (q1) {$q_1$};
  \draw (2,-1) node[draw,shape=circle,minimum size=0.9cm,thick] (q2) {$q_2$};
  
  \draw[->,thick] (q1) to[loop above] node[above] {a,b} (q1);
  \draw[->,thick,dashed] (q2) to[loop below] node[below] {a,b} (q2);
  \draw[->,thick,dashed] (q0) to node[above left] {a} (q1);
  \draw[->,thick,dashed] (q0) to node[below left] {b} (q2);
  
  \draw[fill=black] (-0.5,1.2) circle (0.05cm);
  \draw[->,thick] (-0.5,1.2) -- (q0);
  
  \end{tikzpicture}
};

\draw (0,-0.5) node {$\mathcal{A}_1$};

\node[anchor=south] (a2) at (4.5,0) {
    \begin{tikzpicture}
  \draw (0,0) node[draw,shape=circle,minimum size=0.9cm,thick] (q0) {$q_0$};
  \draw (2,1) node[draw,shape=circle,minimum size=0.9cm,thick] (q1) {$q_1$};
  \draw (2,-1) node[draw,shape=circle,minimum size=0.9cm,thick] (q2) {$q_2$};
  
  \draw[->,thick] (q1) to[loop above] node[above] {a,b} (q1);
  \draw[->,thick,dashed] (q2) to[loop below] node[below] {a,b} (q2);
  \draw[->,thick,dashed] (q0) to node[above left] {a} (q1);
  \draw[->,thick,dashed] (q0) to node[below left] {b} (q2);
  
  \draw[fill=black] (-0.5,1.2) circle (0.05cm);
  \draw[->,thick] (-0.5,1.2) -- (q0);
  
  \end{tikzpicture}
};

\draw (4.5,-0.5) node {$\mathcal{A}_2$};

\end{tikzpicture}
\caption{The automata $\mathcal{A}^1$ and $\mathcal{A}^2$ in a chain of co-Büchi automata $\mathcal{A}^1,\mathcal{A}^2$ for which the number of residual languages in the indivdiual co-Büchi automata is larger than the number of residual languages of the overall language represented by the chain.}
\label{fig:exampleSuffixLanguageBlowup}
\end{figure}

We counter both problems of the naive approach by replacing it with an alternative approach that computes $R^L$ in a step-by-step fashion and that merges states representing the same residual language eagerly. We start by computing residual language tracking automata $R^{L,1}, \ldots, R^{L,n}$ for each history-deterministic co-Büchi automaton of $\mathcal{A}^1, \ldots, \mathcal{A}^n$ (Lemma~\ref{lem:residualLanguageAutomatonForSingleCase}). %
Then, we show how to compute state combinations of automata in the chain that witness different residual languages of $L$ (Lemma~\ref{lem:productPart}). Finally, we show how to employ these witnesses in a process that computes $R^L$ step-by-step, merging equivalent states eagerly and having a running time that is polynomial in the input and output sizes.

We start with computing residual language tracking automata for the individual history-deterministic co-Büchi automata:
\begin{lemma}
\label{lem:residualLanguageAutomatonForSingleCase}
Let $\mathcal{A} = (Q,\Sigma,\delta,q_0,F)$ be a minimized history-deterministic co-Büchi automaton. 
We can compute a residual language tracking automaton $R^L = (S^L,\Sigma,\delta^L,s^L_0)$ for the language of $\mathcal{A}$ and a function $f^{\mathcal{A},R} : Q \rightarrow S^L$ mapping each state of $\mathcal{A}$ to the corresponding state of $R^L$ in polynomial time.
\end{lemma}
\begin{proof}
To build $R^L$, we can employ a result by Kuperberg and Skrzypczak \cite[Theorem 13 in the appendix]{DBLP:conf/icalp/KuperbergS15}. They show that given a history-deterministic co-Büchi automaton $\mathcal{B}$ and some non-deterministic co-Büchi automaton $\mathcal{C}$, one can decide in polynomial time if $\mathcal{L}(\mathcal{B}) \subseteq \mathcal{L}(\mathcal{C})$ by encoding the language inclusion problem in a parity game (with a constant number of colors) and subsequently solving the game. If and only if the initial position in the game is winning for the \emph{duplicator} player, we have $\mathcal{L}(\mathcal{B}) \subseteq \mathcal{L}(\mathcal{C})$. The parity game is of size quadratic in the numbers of states of $\mathcal{A}$ and $\mathcal{B}$, and because for parity game solving over a constant number of colors, we have a polynomial time complexity, we obtain a polynomial computation time.

We can use the procedure to compute an equivalence relation $R$ between states of $\mathcal{A}$ regarding their residual language as two states $q$ and $q'$ have the same residual language if and only if we have $\mathcal{L}(\mathcal{A}_q) \subseteq \mathcal{L}(\mathcal{A}_{q'})$ and $\mathcal{L}(\mathcal{A}_{q'}) \subseteq \mathcal{L}(\mathcal{A}_{q})$. Computing $R$ is also possible in time polynomial in $\mathcal{A}$ because it involves quadratically many invocations of a polynomial-time procedure. Since language equivalence is transitive, reflexive and symmetric, we have that $R$ is an equivalence relation.

Given $R$, we can now define $R^L$ by setting $S^L$ to be the equivalence classes of $R$ (which implies that $|S^L| \leq |Q|$).
Then, we define $q^0_L$ to be class of $R^L$ including $q_0$ and set, for all $c \in S^L, x \in \Sigma$, $\delta^L(c,x) = c'$ to some $c' \in S^L$ such that $\left(\cup_{q \in c} \delta(q,x) \right) \subseteq c'$. 
There is exactly one such class because a minimized history-deterministic co-Büchi automaton is deterministic in its residual language, so all $x$-successors of states in the same residual language have the same residual language.

As mapping $f^{\mathcal{A},R} : Q \rightarrow R^L$, we use a function mapping each state to the equivalence class (w.r.t.~$R$) that it contains.

We note that the parity game construction by Kuperberg and Skrzypczak~\cite[Theorem 13 in the appendix]{DBLP:conf/icalp/KuperbergS15} is for the case of state-based acceptance instead of transition-based acceptance. The construction however also works for the transition-based case if the parity game also uses a winning condition that is transition-based. Furthermore, we note that there is, in practice, no need to build separate parity games for each combination of states $q$ and $q'$, as the games have positions for each combination of states and exactly the positions $(q,q')$ that are winning for the duplicator player in the game have $\mathcal{L}(\mathcal{A}_q) \subseteq \mathcal{L}(\mathcal{A}_{q'})$. So it suffices to build two games and to determine the winning positions in the games in order to obtain $R$.
\end{proof}

\begin{lemma}
\label{lem:productPart}
Let $\mathcal{A}^1, \ldots, \mathcal{A}^n$ be a chain of co-Büchi automata for the language $L$, $0 \leq i < n$ and $0 \leq j < n$ be given, and for notational simplicity $\mathcal{A}^0$ be a single-state co-Büchi automaton for the universal language.
Let furthermore $R^{L,0}, \ldots, R^{L,n}$ be residual language tracking automata for $\mathcal{A}^0, \ldots, \mathcal{A}^n$ with $R^{L,k} = (S^{L,k},\Sigma,\delta^{L,k},s^{L,k}_{0})$ for all $0 \leq j \leq n$.

We can compute, in polynomial time, a set $R^{i,j} \subseteq S^{L,i} \times S^{L,{i+1}} \times S^{L,j} \times S^{L,{j+1}}$ such that if for some finite words $w^1, w^2 \in \Sigma^*$ and 
$s^i \in S^{L,i}$, $s^{i+1} \in S^{L,{i+1}}$, 
$s^j \in S^{L,j}$, and $s^{j+1} \in S^{L,{j+1}}$, 
we have:
\begin{itemize}
\item $\delta^{L,i}(s_{0}^{L,i},w^1) = s^i$,
\item $\delta^{L,{i+1}}(s_{0}^{L,{i+1}},w^1) = s^{i+1}$,
\item $\delta^{L,j}(s_{0}^{L,j},w^2) = s^j$, and
\item $\delta^{L,{j+1}}(s_{0}^{L,j+1},w^2) = s^{j+1}$,
\end{itemize}
then if and only if there exists some infinite word $\tilde w \in \Sigma^\omega$ such that 
\begin{itemize}
\item $w^1 \tilde w \in \mathcal{L}(\mathcal{A}^i)$, 
\item $w^1 \tilde w \notin \mathcal{L}(\mathcal{A}^{i+1})$, 
\item $w^2 \tilde w \in \mathcal{L}(\mathcal{A}^j)$, and
\item $w^2 \tilde w \notin \mathcal{L}(\mathcal{A}^{j+1})$,
\end{itemize}
we have $(s^i,s^{i+1},s^j,s^{j+1}) \in R^{i,j}$.
\end{lemma}

The set of relations $\{R^{i,j}\}_{1 \leq i \leq n, 1 \leq j \leq n}$ computed according to the previous lemma is used by the procedure for computing the residual langauge tracking automaton for the language encoded in a chain of co-Büchi automata given in Figure~\ref{fig:residualComputationAlgorithm}.

\begin{figure*}
\begin{algorithmic}[1]
\Procedure{ResidualAutomatonComputation}{$\mathcal{A}^1, \ldots, \mathcal{A}^n$}
\State compute $R^{L,1}, \ldots, R^{L,n}$ for $\mathcal{A}^1, \ldots, \mathcal{A}^n$ according to Lemma~\ref{lem:residualLanguageAutomatonForSingleCase}
\State compute $\{R^{i,j}_L\}_{1 \leq i \leq n, 1 \leq j \leq n}$ for $\mathcal{A}^1, \ldots, \mathcal{A}^n$ according to
Lemma~\ref{lem:productPart}
\State $S^L \gets \{(s^{L,0}_0, \ldots, s^{L,n}_0)\}$
\State $\delta^L \gets \emptyset$, $s^L_0 = (s^{L,0}_0, \ldots, s^{L,n}_0)$
\State $\mathsf{ToDo} \gets S^L$
\While{$\mathsf{ToDo} \neq \emptyset$}
\State remove some $(s^1, \ldots, s^n)$ from $\mathsf{ToDo}$
\For{$x \in \Sigma$} \label{lin:characterloop}
\State $(s'^1, \ldots, s'^n) \gets (\delta^{L,1}(s^1,x), \ldots, \delta^{L,n}(s^n,x))$
\For{$(s''^1, \ldots, s''^n) \in S^L$} \Comment{Check if there is already a suitable successor state in $S^L$} \label{lin:forLoopFindingOldState}
\If{for no $i,j$ with different evenness, we have $(s'^i,s'^{i+1},s''^j,s''^{j+1}) \in R^{i,j}$}
\State $\delta^L \gets \delta^L \cup \{((s^1, \ldots, s^n),x) \mapsto (s''^1, \ldots, s''^n)\}$
\State continue with next iteration of the loop starting in Line~\ref{lin:characterloop}
\EndIf
\EndFor
\State Add $(s'^1, \ldots, s'^n)$ to both $S^L$ and $\mathsf{ToDo}$ \Comment{No suitable previous equivalent state found} \label{line:addState}
\State $\delta^L \gets \delta^L \cup \{((s^1, \ldots, s^n),x) \mapsto (s'^1, \ldots, s'^n)\}$\label{line:addTransitionToNewState}
\EndFor
\EndWhile
\State return $R^L = (S^L,\Sigma,\delta^L,s^0_L)$
\EndProcedure
\end{algorithmic}
\caption{Algorithm for computing a residual language tracking automaton $R^L$ from a chain of history-deterministic co-Büchi automata $\mathcal{A}^1, \ldots, \mathcal{A}^n$}
\label{fig:residualComputationAlgorithm}
\end{figure*}

We now show that this procedure actually computes $R^L$.

\begin{theorem}
\label{theorem:buildingRL}
Let $\mathcal{A}^1, \ldots, \mathcal{A}^n$ be a chain of history-deterministic co-Büchi automata together representing a language $L$ (by Def.~\ref{def:COCOA}). Given this chain, the algorithm in Figure~\ref{fig:residualComputationAlgorithm} computes the residual language tracking automaton $R^L$ of $L$ in time polynomial in the sum of the numbers of states in $\mathcal{A}^1, \ldots, \mathcal{A}^n$ and $R^L$.
\end{theorem}
\begin{proof}
We note that the algorithm in Figure~\ref{fig:residualComputationAlgorithm} first computes the residual language tracking automata $R^{L,1}, \ldots, R^{L,n}$ for $\mathcal{A}^1, \ldots, \mathcal{A}^n$ and the relations $\{R^{i,j}_L\}_{1 \leq i \leq n, 1 \leq j \leq n}$ according to Lemma~\ref{lem:productPart}. 
Let $f^{L,1}, \ldots, 
 f^{L,n}$ furthermore be the functions mapping the states of $\mathcal{A}^1, \ldots, \mathcal{A}^n$ to the respective states encoding their residual languages.
 
The algorithm gradually builds $R^L$ in a state by state manner, where each state is marked by residual languages for each of $\mathcal{A}^1, \ldots, \mathcal{A}^n$. In particular, we want to show:

\begin{enumerate}
\item the algorithm maintains the invariant that if in $R^L$, we have $\delta^L(s^L_0,w) = (s'^1, \ldots, s'^n)$ for some $w \in \Sigma^*$, then for all $\tilde w \in \Sigma^\omega$, if and only if $w \tilde w \in L$, we have that the highest index $i$ such that $\tilde w$ is accepted by $\mathcal{A}^i_{q^i}$ with $f^{L,i}(q^i) = s'^i$ is even.
\item At no time there exist two states for the same residual language in $R^L$.
\item Eventually, $R^L$ is a complete automaton, i.e., for every finite word $w \in \Sigma^*$, $\delta(s^L_0,w)$ is defined.
\end{enumerate}

We note that the quantification over $q^i$ in the first bullet point is both universal and existential because $R^{L,i}$ tracks residual languages, so that states of $\mathcal{A}^i$ that map to the same $R^{L,i}$ state accept the same words.

The first two claims will be shown by induction, while the third claim follows from the fact that the algorithm maintains a list of states for which no outgoing transitions exist yet. For each of them, the algorithm adds outgoing transitions for each character and at the end of the procedure, no state remains in this list. The three claims together imply the correctness of $R^L$, so it remains to show the first two claims. 

Initially, when $S^L$ consists only of $S^L_0 = \{(s^{L,1}_1, \ldots, s^{L,n}_n)\}$, the second claim trivially holds, and the first claim holds as well, as the initial states of each of $\mathcal{A}^1, \ldots, \mathcal{A}^n$ have the initial residual languages of $R^{L,1}, \ldots, R^{L,n}$. We hence only need to show that no added transition destroys the property that the first two claims hold.

To see this, we consider the possible ways in which the algorithm can add a transition for a character $x$ from some state $(s^1,\ldots,s^n)$: either the loop starting in Line~\ref{lin:forLoopFindingOldState} finds an existing state or not, in which case lines~\ref{line:addState} to \ref{line:addTransitionToNewState} are executed.

In both cases, the algorithm first computes $(s'^1, \ldots, s'^n) = (\delta^{L,1}(s^1,x), \ldots, \delta^{L,n}(s^n,x))$. Under the first claim holding for every $w \in \Sigma$, we know that 
$(s'^1, \ldots, s'^n)$ has the property that for all $\tilde w$ in $\Sigma^\omega$, if and only if the greatest index $1 \leq i \leq n$ such that $\tilde w \in \mathcal{L}(\mathcal{A}^i_{q'^i})$ for some $q'^i \in Q^i$ with $f^{L,i}(q'^i) = s'^i$ is even, we have that $w x \tilde w \in L$.

\looseness-1 Let us now consider the case that $(s'^1, \ldots, s'^n)$ is discarded by the loop starting in Line~\ref{lin:forLoopFindingOldState} finding a suitable alternative state $(s''^1, \ldots, s''^n) \in S^L$.
This can only happen if we do not have $(s'^i,s'^{i+1},s''^j,s''^{j+1}) \in R^{i,j}$ for any $i$ and $j$ with different evenness. But then, by the description of what $R^{i,j}$ represents, this means that there is no word $\tilde w$ that witnesses that $wx$ has a different residual language (w.r.t.~$L$) than any other word $w'$ under which $(s''^1, \ldots, s''^n)$ is reachable in $S^L$. This in turn implies that adding $\{((s^1, \ldots, s^n),x) \mapsto (s''^1, \ldots, s''^n)\} $ to $\delta^L$ maintains both the first and the second invariant.

If however it is found that $(s'^1, \ldots, s'^n)$ is to be added to $S^L$ by executing lines~\ref{line:addState} to \ref{line:addTransitionToNewState}, then this means that for every other state $(s''^1, \ldots, s''^n) \in S^L$, \emph{some} tuple in $(s'^i,s'^{i+1},s''^j,s''^{j+1}) \in R^{i,j}$ has been found. As these tuples witness that $wx$ induces a residual language that is different to the residual languages induced by any word $w'$ under which the state in $R^L$ can be reached \emph{and} by the fact that all states in $S^L$ are reached under words witnessing different residual languages of $L$ (by the second inductive hypothesis), we can conclude that $wx$ induces a new, different residual language and add $(s'^1, \ldots, s'^n)$ and a corresponding transition from $(s^1, \ldots, s^n)$ for $x$ to $\delta^L$.

The overall complexity of  the algorithm is polynomial in the number of states of the input and output automata as computing $R^{L,1}, \ldots, R^{L,n}$ and $\{R^{i,j}\}_{1 \leq i \leq n, 1 \leq j \leq n}$ can be done in polynomial time, the main loop of the algorithm runs once for every state in the output automaton, and the number of steps executed for each state of the output automaton is also polynomial in the sizes of the input automata.
\end{proof}

\subsection{Building the Rerailing automaton}
Let now $L$ be a language, $R^L$ be a residual langauge tracking automaton for $L$, and 
$\mathcal{A}^1, \ldots, \mathcal{A}^n$ be a COCOA that encodes $L$. For notational convenience, we add the automaton $\mathcal{A}^0$ with the universal language.

We first construct an equivalent chain of floating co-Büchi automata $\mathcal{A}'^0, \mathcal{A}'^1, \ldots, \mathcal{A}'^n$ (using the construction from Def.~\ref{def:residualizedAutomaton} with $R^L$). We note that the chain of floating co-Büchi automata still has the same language (by Lemma~\ref{lem:residualizedAutomaton}).

We define a recursive algorithm for building a minimal rerailing automaton $\mathcal{R} = (Q,\Sigma,\delta,q_0)$ from the chain, which is shown in Figure~\ref{fig:mainAlgo}. In base cases of the recursion, states are added to $Q$, and transitions are added to $\delta$ in the cases in which the algorithm recurses, where for notational convenience we assume $Q$ and $\delta$ to be global variables throughout the recursive calls. The names of all rerailing automaton states (i.e., the elements of $Q$) are combinations of states in the floating co-Büchi automata that the algorithm computes and processes.

The algorithm takes as input a floating co-Büchi automaton $\mathcal{A}^C$ with a single accepting strongly connected component as \emph{context} and a color number $i$ to be considered next, and returns a set of states in $Q$ 
that was added to the rerailing automaton to cover the words that are accepted by the supplied floating co-Büchi automaton. To construct $\mathcal{R}$, the algorithm is called once with $(\mathcal{A}'^0,1)$ as parameter values. Afterwards, the initial state is chosen arbitrarily but with the (only) state of $\mathcal{A}'^0$ marked by $L$ as the residual language as initial state.

\begin{figure}
\begin{algorithmic}[1]
\Procedure{RecurseBuild}{$\mathcal{A}^C,i$}
\State $\mathcal{A}^B \gets \mathcal{A}_\emptyset$ \Comment{Automaton without states}
\For{$j \in \{1, \ldots, n\}$ \textbf{in ascending order}} \label{line:firstForLoop}
\If{$j$ is even $\equiv$ $i$ is even}
\State $\mathcal{A^B} \gets \mathcal{A^B} \cup (\mathcal{A}^C \times \mathcal{A}'^j)$
\Else
\State \label{lin:removeSomeSCCs} Remove all states $(q^c,q^i)$ from $\mathcal{A}^B$ with $\mathsf{safe}((q^c,q^i)) \subseteq \mathsf{safe}((q'^c,q'^i))$ for some  state $(q'^c,q'^i)$ of $\mathcal{A}^C \times \mathcal{A}'^j$ with the same $Q^C$ label and the same residual language
\EndIf
\EndFor
\State Minimize $\mathcal{A}^B$ modulo its $\mathcal{A}^C$ labels \label{line:minimize}
\State $\mathit{newStates} \mapsto \emptyset$ \label{line:newStates}
\For{maximal accepting SCC $(Q',\delta')$ in $\mathcal{A}^B$}
\State $\mathcal{A}' \gets \mathcal{A}^B $ but remove all states apart from $Q'$
\State $\mathit{newStates} \gets \mathit{newStates} \cup  \Call{RecurseBuild}{\mathcal{A}',i+1}$

\EndFor
\For{all states $q^C \in Q^C$ not appearing as first state components in a state of a maximal accepting SCC of $\mathcal{A}^B$}
  \State $\mathit{newStates} \gets \mathit{newStates} \cup \{(q^C,\cdot)\}$
\EndFor \label{line:endForTwo}
\State
\label{line:addTransitions} $\delta \gets \delta \cup \{ ((q^C,q),x,(q'^C,q'),i-1) \mid (q^C,q),(q'^C,q') \in \mathit{newStates}, \delta^C(q^C,x)=q'^C,  \nexists ((q^C,q),x,(q'^C,q''),d) \in \delta \textrm{ for some } d \textrm{ and } q'' \}$
\State \textbf{return} $\mathit{newStates}$
\EndProcedure
\end{algorithmic}
\caption{Recursive Algorithm for constructing a minimal rerailing automaton from a chain of floating co-Büchi automata with $\mathcal{A}^C = (Q^C,\Sigma,\delta^C,f^C)$}
\label{fig:mainAlgo}
\end{figure}

Theorem~\ref{thm:centralNestingThm} forms the main idea of the algorithm. 
In every call of \textsc{RecurseBuild}, an automaton $\mathcal{A}^B$ with the properties mentioned in the claim of the theorem is constructed. To do so, the algorithm
takes products of $\mathcal{A}^C$ with all automata $\mathcal{A}'^j$ for $j \in \{1, \ldots, n\}$ in order to ensure that only words are considered that are accepted by the SCC in $\mathcal{A}^C$. This is done by using a simple product construction, which also guarantees that the states in $\mathcal{A}^B$ are $\mathcal{A}^A$-labeled.

When iterating through the automata $\mathcal{A}'^j$, the algorithm distinguishes between even and odd colors. The ones with the same evenness as $i$ are the ones containing accepting SCCs that need to potentially be characterized in the claim of Theorem~\ref{thm:centralNestingThm}. However, if the newly added SCC \emph{only} accepts words that are also accepted with a later automaton in the chain, it is then removed again in the next step.

Afterwards, the algorithm minimizes $\mathcal{A}^B$ modulo its $\mathcal{A}^C$-labeling in order for Theorem~\ref{thm:centralNestingThm} to be applicable. For each remaining maximal SCC in $\mathcal{A}^B$, the algorithm then recurses in order to compute a sufficiently large $\mathcal{A}^B$-labeled state set for $\mathcal{R}$. 
These state sets (for each maximal SCC in $\mathcal{A}^B$) are later connected with transitions of color $i-1$ in Line~\ref{line:addTransitions}, 
but only for state/input combinations that do not have outgoing transitions within the SCCs of $\mathcal{R}$ computed within a recursive call already. The idea behind doing so is that if there is no transition in the SCCs computed within a recursive call, then the run should be allowed to switch between the corresponding SCCs of $\mathcal{R}$ for $\mathcal{A}^C$.

To account for the fact that some state in $\mathcal{A}^C$ may not be a label of any state in $\mathcal{A}^B$, the algorithm also adds states to $Q$ for the missing $\mathcal{A}^C$ states, as otherwise the state set built in the recursive call would not characterize the maximal SCC of $\mathcal{A}^C$.

In the algorithm, we use the following sub-constructions:
\begin{definition}
Let $\mathcal{A}^1 = (Q^1,\Sigma,\delta^1,f^1)$ and $\mathcal{A}^2 = (Q^2,\Sigma,\delta^2,f^2)$ be two floating co-Büchi automata for some residual language tracking automaton $R^L = (Q^L,\Sigma,\delta^L,q^L_0)$. 

We define $\mathcal{A}^1 \times \mathcal{A}^2 = (Q^\times,\Sigma,\delta^\times,f^\times)$ as the floating co-Büchi automaton with $Q^\times = \{(q^1,q^2) \in Q^1 \times Q^2 \mid f^1(q^1) = f^2(q^2) \}$, 
and for every $(q^1,q^2) \in Q^\times$ and $x \in \Sigma$, we have $f^\times((q^1,q^2)) = f^1(q^1)$ and $\delta^\times((q^1,q^2),x)  = (\delta^1(q^1,x),\delta^2(q^2,x))$ whenever $\delta^1(q^1,x)$ and $\delta^2(q^2,x)$ are both defined.

We also define $\mathcal{A}^1 \cup \mathcal{A}^2 = (Q^\cup,\Sigma,\delta^\cup,f^\cup)$ as the floating co-Büchi automaton with $Q^\cup = Q^1 \cup Q^2$ (where we assume, w.l.o.g., $Q^1$ and $Q^2$ to be disjoint sets), 
and for every $q \in Q^\cup$ and $x \in \Sigma$, we have $f^\cup(q) = f^1(q)$ if $	q \in Q^1$ and $f^\cup(q) = f^2(q)$ otherwise, and $\delta^\cup(q,x)  = \delta^1(q,x)$ if defined, $\delta^\cup(q,x)  = \delta^2(q,x)$ if defined, and $\delta^\cup(q,x)$ being undefined otherwise.
\end{definition}

\begin{lemma}
\label{lem:productAndUnionLemma}
Let $\mathcal{A}^1 = (Q^1,\Sigma,\delta^1,f^1)$ and $\mathcal{A}^2 = (Q^2,\Sigma,\delta^2,f^2)$ be two floating co-Büchi automata for some residual language tracking automaton $R^L = (Q^L,\Sigma,\delta^L,q^L_0)$. 

We have that $\mathcal{A}^1 \times \mathcal{A}^2$ accepts the intersection of the languages of $\mathcal{A}^1$ and $ \mathcal{A}^2$ and $\mathcal{A}^1 \cup \mathcal{A}^2$ accepts the union of the languages of $\mathcal{A}^1$ and $\mathcal{A}^2$. Furthermore, $\mathcal{A}^1 \times \mathcal{A}^2$ is an $\mathcal{A}^1$-labeled automaton. If both $\mathcal{A}^1$  and $\mathcal{A}^2$ are $\mathcal{A}^A$-labeled for some other floating automaton $\mathcal{A}^A$, then so is $\mathcal{A}^1 \cup \mathcal{A}^2$.
\end{lemma}
\begin{proof}
For the intersection, assume that a word is accepted by both $\mathcal{A}^1$ and $\mathcal{A}^2$. Then, there exist runs from word indices $k$ and $k'$, respectively. The definition of $\mathcal{A}^1 \times \mathcal{A}^2$ guarantees that then, there is a run from index $\max(k,k')$ in $\mathcal{A}^1 \times \mathcal{A}^2$. Similarly, if a run in $\mathcal{A}^1 \times \mathcal{A}^2$ exists, by mapping it to the first and second state component, we can obtain runs for $\mathcal{A}^1$ and $\mathcal{A}^2$, respectively. To see that the automaton is $\mathcal{A}^A$-labeled, we can use a function mapping states to their first component as labeling function.

For $\mathcal{A}^1 \cup \mathcal{A}^2$, assume that a word is accepted by that floating automaton. Since the SCCs from $\mathcal{A}^1$ and $\mathcal{A}^2$ are separate in $\mathcal{A}^1 \cup \mathcal{A}^2$, one of the automata in $\mathcal{A}^1$ or $\mathcal{A}^2$ accepts the word. On the other hand, runs for $\mathcal{A}^1$ or $\mathcal{A}^2$ are also runs for $\mathcal{A}^1 \cup \mathcal{A}^2$.
If both $\mathcal{A}^1$ and $\mathcal{A}^2$ are $\mathcal{A}^A$-labeled, then we can obtain a labeling function for $\mathcal{A}^1 \cup \mathcal{A}^2$ by merging the labeling functions of $\mathcal{A}^1$ and $\mathcal{A}^2$.
\end{proof}

\begin{theorem}
Running the algorithm in Figure~\ref{fig:mainAlgo} on $(\mathcal{A}^0,1)$ as input for a chain of floating co-Büchi automata for some language $L$ yields a minimal rerailing automaton for $L$.
\label{theorem:mainAlgoCorrect}
\end{theorem}

\begin{proof}
We prove the claim by structural induction. Let $Q'$ be the states generated by the algorithm for some context $\mathcal{A}^C = (Q^C,\Sigma,\delta^C,f^C)$. We prove that 
for the states $(q,q_\mathit{rest}) \in Q'$ with $q \in Q^C$ and $v = v_0 \ldots v_{|v|-1} \in \Sigma^*$ with $q \in \delta^+(q_0,v)$ and some word $\tilde w = w_0 w_1 \ldots \in \Sigma^\omega,$ if and only if we have that $\delta^C(q,w_0 \ldots w_i)$ is defined for all $i \in \NN$, then no run in $\mathcal{R}$ for $vw$ vising $q$ after $|v|$ steps leaves $Q'$ after $|v|$ steps, and $w$ is accepted by $\mathcal{R}_q$ iff $vw$ is in $L$. 

Furthermore, we prove inductively that any call to \textsc{RecurseBuild}$(\mathcal{A}^C,i)$ only adds transitions with a color of at least $i$ to $\mathcal{R}$, and 
that for any word/combination $w$/$q$, if all runs for $w$ from $q$ stay in $Q'$, then the word is accepted by $\mathcal{R}$ if and only if $vw \in L$. This includes showing the rerailing property for such words.

To see that all runs for $w$ are stuck in $Q'$ from $q$, note that all states added in recursive calls track the state of $\mathcal{A}^C$ (by Lemma~\ref{lem:productAndUnionLemma}), and in Line~\ref{line:addTransitions} of the algorithm, whenever for some state $(q,q_\mathit{rest})$ in the automaton, there is not already an outgoing transition for some character $x \in \Sigma$, but $\delta^C(q,x)$ is defined, then transitions to all other states in $Q'$ with $\delta^C(q,x)$ as first component are added, and there exists at least one such state in $Q'$ for each state in $Q^C$. Hence, a run for $w$ cannot leave $Q'$ if $\delta^{C}(q,w_0 \ldots w_i)$ is defined for every $i \in \NN$.
On the other hand, no $x$-transitions from a state $(q,q_\mathit{rest})$ are added if $\delta^C(q,x)$ is undefined. Since the algorithm is called recursively with $\mathcal{A}^B$, which is intersected with $\mathcal{A}^A$, by the induction hypothesis, this property also holds for the transitions within $Q'$ added during the recursive calls. The state set $Q'$ is never left as if $Q'$ is computed by a call to \textsc{RecurseBuild}, no transitions leading out of $Q'$ that can be taken for $w$ (when starting from $q$) are added later in the algorithm.

We now consider two cases (\textbf{A} and \textbf{B}) for words $w$ accepted from some state $q$ in $\mathcal{A}^C$: either the word has a terminal SCC entirely within a sub-SCC $(Q'',\delta'')$ with $Q'' \subseteq Q'$ computed by a recursive call (considering the transitions in $\delta''$ added by the recursive call) or not.

\textbf{(A:)} In the former case, assume that there is some run $\pi$ for $vw$ going through a state $(q,q_\mathit{rest})$ of $\mathcal{R}$ at position $|v|$ and reaching a rerailing point at position $k' > |v|$ in the run within state set $Q''$. Then let $m : \{k',k'+1, \ldots \} \rightarrow 2^{Q''}$ be a mapping showing in which states runs after the rerailing point can be in. Note that by the description of the state space of $Q'$ above, we have that for each $j \geq k'$, the first element of each state in $m(j)$ is $\delta^C(q,w_0 \ldots w_j)$.

By the inductive hypothesis, we have that for all words getting stuck in $Q''$, $\mathcal{R}$ already has the rerailing property and accepts with the right color. 
We hence only need to show that runs $\pi'$ for $vw$ with $\pi_{|v|} = q$ and such that such runs stay in $Q'$ after position $|v|$ can be rerailed to eventually get stuck in $Q''$ \emph{or} some other SCC $Q'''$ build during recursive calls to the algorithm. In both cases, the inductive hypothesis can be applied then. Getting stuck in $Q'''$ instead of $Q''$ can happen if the word under concern is accepted by multiple maximal SCCs of $\mathcal{A}^B$.

If a run neither gets stuck in $Q''$ nor gets stuck in a potentially other existing sub-SCC $Q'''$, then it infinitely often takes the transitions added in Line~\ref{line:addTransitions}, and these transitions have a color of $i-1$. Hence, this run has a lower color than any other run within $Q''$ (or $Q'''$), as by the inductive hypothesis, only colors of at least $i$ can occur in transitions added by a recursive call. At the same time, whenever at some point $j \geq k'$ in the run one of these added transitions is taken, there are other transitions for the same characters to all states in $m(j)$ for this word. Since after taking one of these transitions, a rerailing point has been hit that, by the inductive hypothesis, leads to the inductive claim holding, we have shown the automaton to be correct for this case, and the dominating color of the run  increase (to at least $i$). Also, the inductive claim over the minimum color appearing in an SCC holds because that is $i-1$ as all colors of sub-SCCs are bigger. 

\textbf{(B:)} Now consider the other case, so that every run entering a sub-SCC computed by recursive calls eventually leaves the sub-SCC (as otherwise there would be a terminal SCC in the sub-SCC). Since this means that the added transitions in Line~\ref{line:addTransitions} are taken infinitely often along every run, we have that every run has a color of $i-1$. The steps of the construction of $\mathcal{A}^B$ make sure that this can only happen if $i-1$ is even if and only if the overall word is in $L$. 
To see this, consider that 
$\mathcal{A}^B$ is built in a very specific way in the algorithm.
In particular, the algorithm iterates through the possible chain elements $j$, and whenever $j$'s evenness is the same as the evenness of $i$, SCCs are added to $\mathcal{A}^B$ that accept all words accepted both by $\mathcal{A}'^j$ and $\mathcal{A}^C$. Among the words that can get stuck in $\mathcal{A}^C$, these are the ones for which their saturating words are in $L$ if and only if the evenness of $i$ and $j$ coincide \emph{unless} possibly the respective word is also accepted by some automaton $\mathcal{A}'^{j'}$ with $j'>j$. This is the case for some SCC in $\mathcal{A}^B$ if and only if there is some SCC in $\mathcal{A}'^{j'}$ accepting \emph{all} words accepted with the SCC in $\mathcal{A}^B$. Hence, in Line~\ref{lin:removeSomeSCCs} of the algorithm, the respective SCCs are removed again. If now a word is not accepted by $\mathcal{A}^B$ but by $\mathcal{A}^C$, this means that no $j$ with the same evenness as $i$ is the highest index of an automaton $\mathcal{A}'^j$ accepting the word, and hence the word is in $L$ if and only if $i$ is \emph{not} even. Assigning color $i-1$ to all runs for such words by making the transitions between the SCCs of $Q'$ built with recursive calls then ensures that every run corresponding to the word has color $i-1$, which is a suitable color for the word and it is also the color of a rerailing run within $Q'$.

Minimality (in the number of states) follows by structural induction over applications of Theorem~\ref{thm:centralNestingThm}.
Every call to the main algorithm yields the minimal number of states needed in order to encode a part of $\mathcal{R}$ that characterizes the words accepted by the context $\mathcal{A}^C$. For the base case ($\mathcal{A}^B$ has no state), this follows from the fact that in order to characterize some $\mathcal{A}^C$, at least as many states are needed as there are in $\mathcal{A}^C$ (as $\mathcal{A}^C$ is always minimized), and for the other cases, this claim follows from Theorem~\ref{thm:centralNestingThm} (and using the fact that in order to characterize $\mathcal{A}^C$, for each state of $\mathcal{A}^C$, at least one state in the respective SCC of $\mathcal{R}$ needs to correspond to the state).
As a rerailing automaton tracks residual languages (Lemma~\ref{lem:rerailingAutomatonSuffixLanguageLabeling}), and $\mathcal{A}'^0$ does nothing other than tracking residual languages while being of minimal size among the automata doing so, the correctness of the minimality claim follows.
\end{proof}

We note that the automaton built by the minimization algorithm is \emph{color-homogeneous}, i.e, for every state $q$ and letter $x$, all outgoing $x$-transitions from $q$ have the same color. This is a property that minimized rerailing automata hence share with transition-based co-Büchi automata with history-deterministic acceptance minimized by the corresponding algorithm of Abu Radi and Kupferman~\cite{DBLP:journals/lmcs/RadiK22}.

Finally, we discuss the complexity of the rerailing automaton construction procedure.
\begin{lemma}
The recursive translation procedure in Figure~\ref{fig:mainAlgo} runs in time polynomial in the sizes of the input and output automata.
\end{lemma}
\begin{proof}
We first of all note that the recursion depth of $\mathcal{A}^C$ is at most $n$. In particular, we can show by structural induction that for every call \textsc{RecurseBuild}$(\mathcal{A}^C,i)$, we have that that all words accepted by $\mathcal{A}^C$ are accepted by $\mathcal{A}'^{i-1}$. For the induction basis, this claim is trivial (as executing \textsc{RecurseBuild}$(\mathcal{A}'^0,1)$ is the entry point). For the induction step, assume that the inductive hypothesis is true. Then when iterating through the values of $j$, in Line~\ref{lin:removeSomeSCCs} of the algorithm, $\mathcal{A}^B$ gets an empty state set for $j=i-1$, as all words accepted by the only SCC in $\mathcal{A}^C$ are accepted by $\mathcal{A}'^{j-1}$. In later iterations of the loop starting in Line~\ref{line:firstForLoop}, only SCCs are added that accept words that are accepted both by $\mathcal{A}'^{j-1}$ and $\mathcal{A}'^{j}$. This means that for all later calls to \textsc{RecurseBuild}$(\mathcal{A}',i+1)$, as $\mathcal{A}'$ contains a maximal SCC from $\mathcal{A}^B$, the inductive hypothesis also holds.

We now note that in every call \textsc{RecurseBuild}$(\mathcal{A}^C,i)$, the procedure adds at least as many states to $\mathcal{R}$ as there are states in $\mathcal{A}^C$. Also, the size of $\mathcal{A}^B$ is only at most polynomial in the sizes of $\mathcal{A}^C$ and the sum of sizes of $\mathcal{A}'^1, \ldots, \mathcal{A}'^n$ (as the $\cup$ operation on floating automata does not incur a superlinear blow-up). Minimization of floating automata can also be done in time polynomial in the input size (Lemma~\ref{lem:minimizationLemma}). 

Because for all contexts $\mathcal{A}^C$ at the same $i$ level, the generated states are distinct, within all recursion steps, lines~\ref{line:firstForLoop} to \ref{line:minimize} hence overall take time polynomial in the input size and linear in the output size. Similarly, the number of recursive calls is limited by $n$ times the number of states of the output automaton. Lines~\ref{line:newStates} to \ref{line:endForTwo} overall take time linear in the square of the output automaton's size (as $\mathcal{A}^B$ is limited by that size, but each SCC of it is to be considered). Finally, the lines afterwards computing the states of $\mathcal{A}^C$ not already covered and the transitions of color $i-1$ to be added also depend polynomially on the output automaton size as bound on the computation time.
\end{proof}

We note that due to the first inductive argument of the proof above, as an optimization to the algorithm, instead of iterating over all values $j \in \{1, \ldots, n\}$ in Line~\ref{line:firstForLoop}, we can also iterate over all $j \in \{i, \ldots, n\}$ instead. We chose not to apply this optimization in Figure~\ref{fig:mainAlgo} in order to separate this optimization from the correctness proof in Theorem~\ref{theorem:mainAlgoCorrect}.

While from a chain of co-Büchi automata, the minimized automaton built according to the algorithm in Figure~\ref{fig:mainAlgo} may be of size exponential in the sum of sizes of the automata in the chain, we still have an overall polynomial complexity when applying it in order to minimizing rerailing automata:

\begin{corollary}
Given a rerailing automaton $\mathcal{R}$, let the following steps be executed:
\begin{itemize}
\item Translating $\mathcal{R}$ to a chain of co-Büchi automata $\mathcal{A}^1, \ldots, \mathcal{A}^n$ according to Lemma~\ref{lem:decomposingRerailingAutomatonIntoNonCanonicalCOCOA}
\item Computing the residual language tracking automaton $R^L$ of the chain according to Lemma~\ref{lem:residualLanguageAutomatonForSingleCase}
\item Computing a corresponding chain of floating co-Büchi automata $\mathcal{A}'^1, \ldots, \mathcal{A}'^n$ from  $\mathcal{A}^1, \ldots, \mathcal{A}^n$ and $R^L$ according to Definition~\ref{def:residualizedAutomaton} and Lemma~\ref{lem:residualizedAutomaton}
\item Computing a minimal rerailing automaton $\mathcal{R}'$ from $\mathcal{A}'^1, \ldots, \mathcal{A}'^n$ and $R^L$ according to the construction from Figure~\ref{fig:mainAlgo}.
\end{itemize}
We have that $\mathcal{R}'$ is a minimized automaton of $\mathcal{R}$, and all steps together take time polynomial in the size of $\mathcal{R}$.
\end{corollary}
\begin{proof}
Language equivalence of $\mathcal{R}$ and $\mathcal{R}'$ follows from the properties of the respective procedures.

We can then use the fact that rerailing automata track residual languages, so we know that the size of $R^L$ is not larger than the size of $\mathcal{R}$. Furthermore, the size of $\mathcal{R}'$ is not larger than the size of $\mathcal{R}$ as the former is guaranteed to be minimal. Taking the complexities of the steps together and exploiting the fact that the output of the whole procedure is at most as big as $\mathcal{R}$, we obtain the stated complexity bound.
\end{proof}

\section{Application: Realizability Checking for Reactive Systems}

In this section, we show that rerailing automata can be used to check the \emph{realizability} of a linear-time specification. Given a language $L \subseteq \Sigma^\omega$ over some alphabet $\Sigma = \Sigma^I \times \Sigma^O$, its realizability problem \cite{DBLP:reference/mc/BloemCJ18} is to decide if a there exists a function $f : (\Sigma^I)^* \rightarrow \Sigma^O$ such that for all words $w=(w^I_0,w^O_0) (w^I_1,w^O_1) \ldots \in \Sigma^\omega$ with $w^O_i = f(w^I_0 \ldots w^I_{i-1})$ for all $i \in \NN$, we have $w \in L$.

A classical approach to solving the realizability problem is to represent the language $L$ as a deterministic parity automaton and to translate the automaton to a parity game between two players. The \emph{environment} player makes choices from $\Sigma^I$ and the system player makes choices from $\Sigma^O$. If and only if the system player has a winning strategy to win the parity game, the specification is realizable. Determining the winner of the game can then be done with one of the known parity game solving algorithms. Building the game is relatively straightforward and amounts to splitting the transitions in the parity automaton into $\Sigma^I$ and $\Sigma^O$ components. 

We show here that when a specification is given as a rerailing automaton that was minimized with the construction from the previous section, we can construct such a parity game for the same purpose in a similar way. This allows taking advantage of polynomial-time minimization prior to building a parity game.

\begin{theorem}
\label{thm:synthesis}
Let $\mathcal{R} = (Q,\Sigma,\delta,q_0)$ be a rerailing automaton with $\Sigma = \Sigma^I \times \Sigma^O$ that is the outcome of the procedure from the previous section.
Let furthermore for each $(q,(x,y)) \in Q \times \Sigma$ and $c \in \NN$ be $f_c(q,(x,y)) = \{q' \in Q \mid (q,(x,y),q',c) \in \delta\}$.

We construct a parity game $\mathcal{G} = (V^0,V^0,E^0,E^1,C,v^0)$ with the position sets $V^0 = Q \cup \bigcup_{\text{odd } c, (q,z) \in Q \times \Sigma, f_c(q,z) \neq \emptyset} \{ (f_c(q,z),c) \}$ and $V^1 = Q \times \Sigma^O \cup \bigcup_{\text{even } c, (q,z) \in Q \times \Sigma, f_c(q,z) \neq \emptyset} \{ (f_c(q,z),c) \}$, the edges
\begin{align*}
E^0 & = \{(q,(q,y)) \mid q \in Q, y \in \Sigma^O \} \\
& \cup \{ ((Q',c),q) \in 2^Q \times \NN \times Q \mid q \in Q'  \} \\
E^1 & = \{((q,y),(Q',c)) \mid q \in Q, y \in \Sigma^O, \exists x \in \Sigma^I: \\
 & \quad \exists (q,y,q',c) \in \delta.\ Q' = f_c(q,(x,y)) \} \\
 & \  \cup \{ ((Q',c),q) \in 2^Q \times \NN \times Q \mid q \in Q'\},
\end{align*}
the coloring function $C$ mapping each position of the shape $(Q',c)$ to color $c$ and all other positions to the biggest color occurring along some transition in $\mathcal{R}$, and the initial vertex $v^0 = q_0$.

We have that player $0$ (the \emph{system player}) has a stategy to win the game from $v^0$ iff the specification given by $\mathcal{R}$ is realizable for the language split $\Sigma = \Sigma^I \times \Sigma^O$.
\end{theorem}
\begin{proof}
Since $\mathcal{R}$ represents an omega-regular language, we can use the fact that either a function $f$ proving the specification to be realizable exists, or there exists an alternative (\emph{counter-strategy}) function $f' : (\Sigma^O)^+ \rightarrow \Sigma^I$ such that for all  words $w=(w^I_0,w^O_0) (w^I_1,w^O_1) \ldots \in \Sigma^\omega$ that have $w^I_i = f(w^O_0 \ldots w^O_{i})$ for all $i \in \NN$, we have $w \notin L$.

We show next that if $f$ exists, there is a winning strategy for the system player from $v^0$. The case of showing that if $f'$ exists, then the environment player has a winning strategy from $v^0$ is then analogous.

Let $\pi = \pi_0 \pi_1 \ldots \in (V^0 \cup V^1)^\omega$ be a play through $\mathcal{G}$. We note that exactly every third element of $\pi$ is of the form $(Q',c)$ for some $Q' \subseteq Q$.
We define $\pi|_\Sigma = (\rho^O_0,\rho^I_0) (\rho^O_1,\rho^I_1) \ldots$ to be the mapping of $\pi$ to the elements of $\Sigma^I$ and $\Sigma^O$ involved in the play. These are, for every $i \in \NN$, the element $y \in \Sigma^O$ such that $\pi_{3i+1} = (q,y)$ for some $q \in Q$, alternating with, for every $i \in \NN$, the element $x \in \Sigma^I$ such that $\pi_{3i+2} = (f_c(q,(x,y)),c)$. 

The strategy for the system player follows $f$, so that for all plays $\pi$ resulting from the strategy, we have $\pi|_\Sigma \in \mathcal{L}(\mathcal{R})$. This fixes the behavior of the stategy in some vertices, but not in the vertices of the form $(Q',c)$. In such a case, the system player also needs to make a choice. 

Let for each $q$ and $i \in \NN$ be $\mathsf{IncomingRuns}_i(q)$ the set of color sequence/run combinations $\theta$/$\pi'$ in $\mathcal{R}$ with $\theta = \theta_0 \ldots \theta_{i-1} \in \NN^*$ while $\pi' = \pi'_0 \ldots \pi'_i$ is some run with $\pi'_0 = q_0$, $\pi'_i = q$, and for all $0 \leq j < i$, we have $(\pi'_j,(\rho^O_j,\rho^I_j),\pi'_{j+1},\theta_j) \in \delta$.
Whenever at a position $3\cdot i + 2$ in a play with $\pi_{3 \cdot i +2} = (Q',c)$ for some $Q' \subseteq Q$ and $c \in \NN$, the system player has a choice, we let it take a transition to the vertex $q$ that maximizes $\max \{ k \in \NN \mid \exists \theta_0 \ldots \theta_{i-1} / \pi'_0 \ldots \pi'_{i} \in \mathsf{IncomingRuns}_i(q). \forall i-k \leq j < i-1.\  \theta_j \leq c \}$.
In other words, the system player always takes the choice corresponding to the (available) state that has a run for the input/output observed along the play so far that has not seen a transition with a color of at most $c$ as long as possible.

To see that every resulting play is winning for the system player, consider the converse. Then, there is some (odd) color $c$ that is the dominating color of $\pi$. 
Consequently, there are infinitely often positions $i \in \NN$ with $\pi_{3i+2} = (Q',c)$ for some $Q' \subseteq Q$ along the play (using the fact that $\mathcal{R}$ is color-homogeneous). 
By the definition of the Algorithm in Figure~\ref{fig:mainAlgo}, we know that all these color $c$ transitions happening infinitely often along the play are part of a strongly connected component $Q''$ in $\mathcal{R}$ built during a call to \textsc{RecurseBuild}$(\mathcal{A}^C,c+1)$ for some $\mathcal{A}^C$. This is because transitions connecting such SCCs have a color smaller than $c$. 
If $c$ is the dominating color of the run (as assumed), then eventually all transitions of $\mathcal{R}$ taken along the play are within $Q''$, say from position $k'$ in the play. As $Q''$ characterizes words accepted by $\mathcal{A}^C$ (by the central nesting theorem), this means that all runs for $\pi|_\Sigma$ stay in $Q''$. 

This means that if $\pi|_\Sigma$ is accepted by $\mathcal{R}$, then there exists a rerailing run within $Q''$ that recognizes the word with a color greater than $c$. Since all states in $Q''$ track $\mathcal{A}^C$ states (by the central rerailing theorem), and since the construction in the main algorithm ensures that whenever a transition with color $c$ is taken, all states in $Q''$ corresponding to the target $\mathcal{A}^C$ state are reachable, we have that each time a color $c$ transition is taken, the system player can choose to reach a rerailing point with a single step.

If the system player follows the strategy defined above at all positions $3 \cdot i +2$, there are three possible outcomes. One is that eventually no more position with color $c$ is visited, contradicting the assumption that $c$ is the dominating color of the run. The other one is that eventually an edge to a position representing a rerailing point in the respective run in $\mathcal{R}$ is taken, after which color $c$ is not visited any more, again contradicting that $c$ is the dominating color of the play. The remaining case is that 
despite that following the strategy, color $c$ positions are visited infinitely often. To see that this cannot happen, 
observe that once the terminal rerailing point has been reached for a run of $\mathcal{R}$ for $\pi|_\Sigma$ within $Q''$, $\max \{ k \in \NN \mid \exists \theta_0 \ldots \theta_{i-1} / \pi'_0 \ldots \pi'_{i} \in \mathsf{IncomingRuns}_i(q). \forall i-k \leq j < i-1.\  \theta_j \leq c \}$ for the state $q$ by which the vertex is marked grows by $1$ with every three steps of the play as there is a continuous run in $\pi''$ in $\mathcal{R}$ on which color $c$ transitions are not taken any more such that eventually, along $\pi$, an edge to the respective state in $\pi''$ is taken. This run $\pi''$ is in a sub-SCC $Q'''$ of $Q''$. By the construction of the minimization algorithm, the states in $Q'''$ are labeled by states of $\mathcal{A}^B$ and no transition with color $c$ is taken from there if and only if the remainder of $\pi|_\Sigma$ has an accepting run of $\mathcal{A}^B$ from the state of $\mathcal{A}^B$ by which the state in $Q'''$ is marked. This means that if along \emph{some} run within $Q'''$, eventually no transition with color $c$ is taken any more, this applies to \emph{all} runs for the same word. But then, an edge with color $c$ (or lower) is also not visited along the play any more at that point, contradicting the assumption that $c$ is the dominating color of the play.

Overall, we arrive at a contradiction with the assumption that the color of the play is $c$ for some odd $c$.

\end{proof}

\section{Related work}

While rerailing automata as introduced here are a new representation for $\omega$-regular languages, they are related to several other models for $\omega$-regular languages.

In this context, deterministic automata over finite words (DFAs) representing lasso languages \cite{DBLP:conf/mfps/CalbrixNP93} stand out as a model that is polynomial-time minimizable. In a nutshell, the idea is to represent the ultimately periodic words $uv^\omega$ in a given $\omega$-language as a DFA accepting words of the form $u\$v$ for some fixed character $\$$ not used in the alphabet of the $\omega$-language. The corresponding DFA can be minimized using classical MyHill/Nerode equivalence. As two $\omega$-regular languages are equivalent if and only if they contain the same set of ultimately periodic words, such DFAs uniquely represent a given $\omega$-language.

Automata over lasso languages are however difficult to employ in applications without first translating them back into some possibly non-minimized model, as the words accepted by such automata are represented in a rather indirect way.
They have been further refined to \emph{families of deterministic automata over finite words} \cite{DBLP:journals/lmcs/AngluinBF18}, which split the lasso language to be represented into a DFA for the stem $u$ of a word and one DFA for the period $v$ for each state in the stem automaton. They come in several variants and are interesting due to their suitability for \emph{language learning} \cite{DBLP:conf/alt/AngluinF14}.

Chains of co-Büchi automata \cite{DBLP:conf/fsttcs/EhlersS22} have recently been introduced as a model for arbitrary omega-regular languages that can be minimized in polynomial time. Its key idea is a definition of a canonical split of $\omega$-regular languages into a sequence of co-Büchi languages according to the \emph{natural colors} of all words (with respect to the given language), and the co-Büchi languages can be represented by polynomial-time minimizable history-deterministic co-Büchi automata with transition-based acceptance. Such a chain is a possible form of input to our construction for building a minimal rerailing automaton for a given language in this paper. However, our construction does not rely on the language being represented in the canonical way. The precise relationship between the dominating colors of runs of the automata built by the algorithm in Figure~\ref{fig:mainAlgo} in this paper and the natural colors of the respective words is not explored in this paper and left for future work.

Rerailing automata build on parity acceptance. Niwinski and Walukiewicz \cite{DBLP:conf/stacs/NiwinskiW98} defined \emph{flowers} for parity automata that prove that a parity automaton for a given language needs to have a certain number of colors. These flowers in a sense represent a nesting of strongly connected components with transitions of only a certain set of colors that every automaton for a given language needs to have. The central nesting theorem in this paper shows that SCCs in rerailing automata for a language also need to be nested in a certain way. The concept of flowers has recently been extended to \emph{semantic flowers} of a language, which are defined in a way that is agnostic to a specific automaton representing some language of interest \cite{DBLP:journals/ipl/DellErbaSTZ24}.

Rerailing automata extend history-deterministic branching in a very specific way in order to allow for polynomial-time minimization.
Other extensions to history-deterministic branching have been proposed in the past. Boker et al.~\cite{DBLP:conf/fsttcs/BokerKLS20} for instance consider \emph{alternating good-for-games automata}, which combine non-deterministic and universal branching in a way that requires both types of branching to be resolvable by players in a game. The construction from Theorem~\ref{thm:synthesis} above assigns resolving the non-determinism in rerailing automata to two different players, which shows that conceptually, rerailing automata and alternating good-for-games parity automata share similarities. Identifying the exact relationship between these models is left for future work. Good-for-games acceptance was furthermore extended to \emph{good-for-MDP acceptance} by Hahn et al.~\cite{DBLP:conf/atva/HahnPSSTW22a}. Finally, \emph{explorable parity automata}~\cite{DBLP:journals/corr/abs-2410-23187} generalize history-determinism/good-for-gamesness (which are the same concept in case of parity acceptance) to a model in which multiple runs of an automaton are considered at the same time, allowing automata of this type to be more concise.

Canonical chains of co-Büchi automata have recently been used for representing the specification in a realizability checking process \cite{DBLP:conf/tacas/EhlersK24}. Here, such a chain is translated to some special form of alternating good-for-games automaton, which is in turn translated to a system of fixpoint equations to be evaluated over a symbolic game graph. Rerailing automata have a structure that is similar to this special product. However, rerailing automata guarantee to never be bigger than the smallest deterministic parity automaton for the same language, which is a property that the special product employed previously for realizability checking does not have. Integrating rerailing automata into this existing framework in order to make symbolic realizability checking over game graphs with a complex structure more efficient is left for future work.

\section{Conclusion \& Outlook}

This paper introduced \emph{rerailing automata}, a polynomial-time minimizable representation for $\omega$-regular languages. The main result of this paper is a precise characterization of which information strongly connected components in a rerailing automaton need to track. Moreover, rerailing automata have been defined in a way such that states whose existence is implied by the characterization theorem suffice for building an automaton for a given language, which enables polynomial-time minimization of this automaton type.

History-deterministic co-Büchi automata with transition-based acceptance are a special case of rerailing automata, as are deterministic parity automata (with transition-based acceptance). Hence, minimized rerailing automata are never bigger than automata of these inherited types. 
Since deterministic parity automata are a special case of rerailing automata, the \emph{central nesting theorem} that characterizes which states are needed in a rerailing automaton is also applicable to them and provides some insight into the structure of deterministic parity automata.

We also showed how the realizability problem of a specification given in the form of a rerailing automaton can be reduced to parity game solving in a way that is very similar to using deterministic parity automata for the same application. 

As the purpose of this paper is to introduce rerailing automata, many open questions and future work directions remain. 
For instance, we did not explore further applications of rerailing automata. In particular, we did not yet exploit the fact that when resolving the non-determinism in minimized rerailing automata uniformly at random, rerailing automata guarantee that the probability of a run to have an even dominating color is $1$ if the word is in the language of the automaton, and this probability is $0$ otherwise. This fact is likely to be useful in the context of probabilistic model checking.

Then, this paper does not contain a precise conciseness characterization of the relationship between the sizes of deterministic parity automata and rerailing automata (for the same language). While the property that rerailing automata can be exponentially more concise than deterministic parity automata is inherited from the respective result for history-deterministic co-Büchi automata~\cite{DBLP:conf/icalp/KuperbergS15}, the capability of rerailing automata to represent languages needing multiple colors in a deterministic parity automaton may potentially allow to strengthen this result.

Then, the construction from Theorem~\ref{theorem:buildingRL} for obtaining a residual language tracking automaton for a chain of co-Büchi automata is relatively clunky. It was employed in this paper because it allowed starting from arbitrary (non-canonical) chains of co-Büchi automata as language representation. By characterizing the additional properties of chains of co-Büchi automata extracted from rerailing automata using Lemma~\ref{lem:decomposingRerailingAutomatonIntoNonCanonicalCOCOA}  or of canonical chains of co-Büchi automata obtained from existing deterministic parity automata~\cite{DBLP:conf/fsttcs/EhlersS22}, we can potentially replace it to improve the performance of computing a minimized rerailing automaton in practical applications. Also, it may be possible to extract canonical chains of co-Büchi automata from a rerailing automaton directly by merging the ideas of the construction in Lemma~\ref{lem:decomposingRerailingAutomatonIntoNonCanonicalCOCOA} in this paper with the existing algorithm for computing a canonical chain of co-Büchi automata from a deterministic parity automaton~\cite{DBLP:conf/fsttcs/EhlersS22}.

Finally, we note that just like minimized history-deterministic co-Büchi automata with transition-based acceptance have additional exploitable properties that non-minimized such automata do not have in general (and which are explained in Section~\ref{sec:minimalHdTbcBWProperties}), this is also the case for rerailing automata. In particular, we used the \emph{color-homogenity} of minimized rerailing automata in the proof of Theorem~\ref{thm:synthesis}. Identifying if there are further specific properties of minimized rerailing automata that can be exploited in practical applications is also left for future work.

\section*{Acknowledgements}
This work was partially funded by the Volkswagen Foundation through the ``Formal Engineering Support for Field-programmable Gate Arrays'' \emph{Momentum} initiative.

\bibliographystyle{plain}
\bibliography{bib}

\appendix

\subsection{Proof of Proposition~\ref{proposition:Parity}}
\begin{proof}
Let a word $w$ be given in $\mathcal{L}_\mathit{parity}(\mathcal{A})$. Then by the properties described in Section~\ref{sec:minimalHdTbcBWProperties}, we have that there exists a run $\pi'' = \pi''_0 \pi''_1 \ldots$ for $w$ that eventually only takes accepting and deterministic transitions. Now given some run $\pi$ for $w$, either it eventually only takes accepting transitions (which are deterministic), and then $\pi'=\pi$ works as rerailing run for showing the rerailing property, \emph{or} when a rejecting transition is taken at some index $i'$ after the index $j$ at which $\pi''$ took its last rejecting transition, then there is a transition available to $\pi''_{i'+1}$, and taking it lets the rest of the run stay deterministically in the accepting transitions. So in the case of $w$ being in $\mathcal{L}_\mathit{parity}(\mathcal{A})$, the rerailing property holds. Furthermore, the dominating color of the resulting run is $2$, inducing that $w$ is in $\mathcal{L}_\mathit{rerail}(\mathcal{A})$. Note that the dominating color of the rerailing run can only be higher than the one for $\pi$, as required for a rerailing automaton.

Let now $w \notin \mathcal{L}_\mathit{parity}(\mathcal{A})$. In this case, \emph{all} runs of $w$ in $\mathcal{A}$ have a dominating color of $1$. Then the rerailing property is satisfied with $\pi' = \pi$ for arbitrary runs $\pi$ and indices $i$. Furthermore, as the dominating runs of the word uniformly have color $1$, we have that $w \notin \mathcal{L}_\mathit{rerail}(\mathcal{A})$. Again, the color of the rerailing run is not higher (as $\pi'=\pi$).
\end{proof}

\subsection{Proof of Lemma~\ref{lem:residualizedAutomaton}}
\begin{proof}
If $R^L$ tracks the residual language of $\mathcal{A}^i$, i.e., prefix words with different residual languages w.r.t~$\mathcal{A}^i$ have runs leading to different states in $R^L$, then it can be seen that $\mathcal{A}^i$ accepts exactly those words accepted by $\mathcal{A}'^i$, as accepting runs for $\mathcal{A}^i$ can be translated to accepting runs in $\mathcal{A}'^i$ and vice versa.
To translate an accepting run of $\mathcal{A}^i$ to $\mathcal{A}'^i$, we just truncate the prefix so that only accepting transitions are left in the run. Additionally, we add the $R^L$ states for the residual languages observed along the trace as second elements -- the definition of runs in floating co-Büchi automata makes the resulting run valid.
Translating a run $\pi_k \pi_{k+1} \ldots$ from some index $k$ in a word $w = w_0 w_1 \ldots$ in $\mathcal{A}'_i$ is slightly more complicated. We take as run for $\mathcal{A}^i$ an arbitrary run for $w_0 \ldots w_{k-1}$. Then we let the run continue arbitrarily (for $w_k w_{k+1} \ldots$) until possibly some rejecting transition is taken at index $k'>k$. In that case, as minimized \HdTbcBW{} guarantee that there are transitions to all language-equivalent states in $\mathcal{A}^i$, including $\pi_{k'+1}$, we attach $\pi_{k'+1} \pi_{k'+2} \ldots$ to the prefix run in $\mathcal{A}^i$ to obtain an accepting run. If however no such index $k'$ can be found, then the run that we started with is already accepting in $\mathcal{A}^i$, which is also sufficient for the claim.

The potentially problematic case to analyze is the residual language tracking automaton for the language of some automaton $\mathcal{A}^i$ not being a factor of $R^L$, i.e., such that for $R^L = (Q^L,\Sigma,\delta^L,q^L_0)$ and some residual language tracking automaton $R^{L,i} = (Q^{L,i},\Sigma,\delta^{L,i},q^{L,i}_0)$ for $\mathcal{A}^i$, there is no mapping $m : Q^L \rightarrow Q^{L,i}$ such that $m(q^L_0)=q^{L,i}_0$ and for all $x \in \Sigma$ and $q \in Q^L$, we have $\delta^{L,i}(m(q),x)=m(\delta^L(q,x))$.

Let $w_0 \ldots w_n$ and $w'_0 \ldots w'_m$ be two prefix words with different residual languages in some automaton $\mathcal{A}^i$, so for some $\tilde w \in \Sigma^\omega$, we have $w_0 \ldots w_n \tilde w \in \mathcal{L}(\mathcal{A}^i)$ but $w'_0 \ldots w'_m \tilde w \notin \mathcal{L}(\mathcal{A}^i)$. If $w_0 \ldots w_n$ and $w'_0 \ldots w'_m$ have the same residual language regarding $L$, then this means that $w_0 \ldots w_n \tilde w$ and $w'_0 \ldots w'_n \tilde w$ have different colors (i.e., the greatest indices of the automaton accepting the words), but these colors are equivalent regarding evenness. Let, w.l.o.g, $w_0 \ldots w_n$ be chosen such that the color of $w_0 \ldots w_n \tilde w$ is $j \in \NN$, which is the greatest among the words $w_0 \ldots w_n$ that have the same residual language (w.r.t.~$L$) as $w'_0 \ldots w'_m$. In this case, we have that $w_0 \ldots w_n \tilde w \in \mathcal{L}(\mathcal{A}'_j)$ but $w_0 \ldots w_n \tilde w \notin \mathcal{L}(\mathcal{A}'_{j+1})$ (if $j<n$), and since the words $w_0 \ldots w_n$ and $w'_0 \ldots w'_m$ have the same residual language, the same states in $\mathcal{A}'^j$ and $\mathcal{A}'^{j+1}$ can be used as starting point for a run starting with $\tilde w$, we also have $w'_0 \ldots w'_m \tilde w \in \mathcal{L}(\mathcal{A}'_j)$ and $w'_0 \ldots w'_m \tilde w \notin \mathcal{L}(\mathcal{A}'_{j+1})$ (if $j<n$). Hence, moving to a chain of floating co-Büchi automata did not change language containment of either word of $w_0 \ldots w_n \tilde w$ and $w'_0 \ldots w'_m \tilde w$ in the language represented by the chain. 
\end{proof}

\subsection{Proof of Lemma~\ref{lem:minimizationLemma}}
\begin{proof}
Minimization can be performed with the same ideas as used by Abu Radi and Kupferman~\cite{DBLP:journals/lmcs/RadiK22} for the case of history-deterministic transition-based co-Büchi automata.

Let $R^L = (Q^L,\Sigma,\delta^L,q^L_0)$. For each state $q \in Q$, we define $\mathsf{Safe}(q)$ to be the safely accepted words (the \emph{safe language}) from $q$, i.e., the set of finite words $w \in \Sigma^*$ for which $\delta(q,w)$ is defined.

We perform the following steps repeatedly on $\mathcal{A}$ until they cannot be applied any longer:
\begin{enumerate}
\item We remove states that are not part of any strongly connected component.
\item Then, whenever there are two states $q,q'$ with $f(q)=f(q')$ and $m(q)=m(q')$ that are part of different maximal strongly connected components of $\mathcal{A}$ and for which $\mathsf{Safe}(q) \subset \mathsf{Safe}(q')$, we remove $q$ from $Q$ and all transitions to and from $q$. 
\item Then, we merge all states $q,q'$ with $f(q)=f(q')$, $m(q)=m(q')$ and $\mathsf{Safe}(q) = \mathsf{Safe}(q')$, using either all outgoing transitions for $q$ or $q'$ for the merged state arbitrarily. After one such merge, we immediately continue with the first step.
\end{enumerate}
None of these steps changes the language of $\mathcal{A}$ or changes the marking along runs that are still present in the minimized automaton:
\begin{enumerate}
\item For the first step, if for some word there is an accepting run, it either never takes the transition, or it takes the transition exactly once. In the latter case, the run can simply start at a later index, as by the definition of $f$, the states track the residual language.
\item For the second step, if there is a run through $q$ at index $k$, we can replace it by a run starting at state $q'$ at index $k$. Since $q$ and $q'$ are part of different maximal strongly connected components, the new run stays valid. By the fact that the states have the same marking, and we have that 
for each $q \in Q$ and $w \in \Sigma^*$, we have $\delta^R(m(q),w) = m(\delta(q,w))$ (phrased as \emph{the marking progresses deterministically henceforth}), all future markings in accepting runs are also the same.
\item If we have an infinite run through $q$, then by merging the states, the $\mathsf{Safe}$ language of the new state is the same as the old states. So runs for the unmerged states can be translated to a run for the automaton with merged states, and by the markings being the same and the markings progressing deterministically, the markings along the runs are also unchanged. Similarly, every run in the automaton with merged states can be translated to a run for the original automaton with the same markings.
\end{enumerate}
Note that in the first two cases, we do not have to show that the language does not grow with the minimization steps as we remove components from the automaton, so the set of accepted word/run/marking combinations can only shrink.

To show that the automaton is minimal among the automata with the same properties listed in the claim, consider another automaton that has strictly fewer states. Assume that it has also been processed by the procedure above (even though it was already smaller to begin with), which can remove transitions between accepting SCCs, which are superfluous. Let the resulting automaton be called $\mathcal{A}^B = (Q^B,\Sigma,\delta^B,f^B)$ with marking function $m^B$.

\textbf{Establishing a fact:} We first note that no state $q$ of $\mathcal{A}$ can have a safe language that is not a subset of the safe language of some other state $q^B$ in $\mathcal{A}^B$ with $f(q) = f^B(q^B)$ and $m(q) = m^B(q^B)$. To see this, assume the converse. Then, for the states $\{q_1, \ldots, q_n\}$ of $\mathcal{A}^B$ with the same $R^L$ labeling as $q$ and the same $m$ labeling as $q$, there exists a sequence of words $w'^1\, \ldots, w'^n$  such that for each $1 \leq i \leq $n, we have that $w'^i$ is not in the safe language of $q_i$, but of $q$.
We now build a word accepted from $q$, but that is not in the language of $\mathcal{A}^B$, which provides a contradiction. 
Let $Q' \subseteq Q$ be the maximal SCC of $q$. We define two helper functions. 
First, we define $g : \{q_1, \ldots, q_n\} \rightarrow \Sigma^*$ such that $g(q_i) = w'_i$ for each $1 \leq i \leq n$. Then, $h$ maps each state of $Q'$ to a word leading from the state back to $q$. We define
\begin{align*}
w^1 & = w'_1 h(\delta(q,w'_1))
\end{align*}
and set, for all $2 \leq i \leq n$:
\begin{align*}
w^i & = \begin{cases} w^{i-1} g(\delta^B(q_i,w^{i-1})) h(\delta(q,\\
\quad w^{i-1} g(\delta^B(q_i,w^{i-1}))) \!\! %
 & \!\!\!\!\!\!\text{if } \delta^B(q_i,w^{i-1}) \text{ is defined}\\
 w^{i-1} & \!\!\!\!\!\!\text{else}
 \end{cases}
\end{align*}
Then, the word $w^\mathit{pre} (w^n)^\omega$ for some $w^\mathit{pre}$ with $\delta^L(q^L_0,w^\mathit{pre}) = f(q)$ shows that $\mathcal{A}$ and $\mathcal{A}^A$ do not accept the same language. This is because, for each $1 \leq i \leq n$, the word $w^i$ is constructed to not be safely accepted from any of $q_1, \ldots, q_i$, as it consists of a prefix $w^{i-1}$ ensuring that for the states $q_1 \ldots q_{i-1}$ already, followed by a word that is not accepted from the state reached in $\mathcal{A}^B$ from $q_i$ by $w^{i-1}$ (if existing).

\textbf{Using the fact established above:} This argument also works with $\mathcal{A}$ and $\mathcal{A}^B$ swapped, which shows that for each $q^L \in Q^L$ and $q^R \in Q^R$, the maximal elements (by set inclusion) of $M_{q^L,q^R} = \{\mathsf{safe}(q) \mid q \in Q, f(q)=q^L, m(q)=q^R\}$ and $M^B_{q^L,q^R} = \{\mathsf{safe}(q^B) \mid q^B \in Q^B, f^B(q)=q^L, m^B(q)=q^R \}$ are the same.

Let us now (finally) show that $\mathcal{A}$ is minimal. To see this, consider that for each maximal SCC $(Q',\delta')$ of $\mathcal{A}$ and each $q^L \in Q^L$ and $q^R \in Q^R$, if there is any $q^L$-labeled and $q^R$-labeled state in $Q'$, then there is some $q^L$-labeled and $q^R$-labeled state $q \in Q'$ with $\mathsf{safe}(q) \in M_{q^L,q^R}$, as otherwise $q$ would have been removed by the second reduction rule. 
Let $\{q_1, \ldots, q_m\}$ be a selection of $q^L$-labeled and $q^R$-labeled states with their safe languages in some set $M_{q^L,q^R}$, where we pick one state for each of the $m$ SCCs in $\mathcal{A}$.

For some state $q_i \in \{q_1, \ldots, q_m\}$, let $q^B$ be a state from $q^B$ with the same safe language and the same labels. We note that for all $w \in \Sigma^*$, if and only if $q' = \delta(q,w)$ is defined, so is $q'^B = \delta^B(q^B,w)$ (because otherwise $q$ and $q^B$ would have different safe languages), and the safe languages of $q'$ and $q'^B$ need to be the same. Also, the markings induced by $f$/$f^B$ and $m$/$m^B$ reachable under some word $w$ need to be the same, as both the marking progresses deterministically.

From all SCCs with this property, we have that $(Q',\delta')$ is the minimal one, as all states with the same safe language and the same markings have been merged. Also, the fact that for all $w \in \Sigma^*$, if and only if $q' = \delta(q,w)$ is defined, so is $q'^B = \delta^B(q^B,w)$ implies that all other states $q_j$ with $i \neq j$ and $1 \leq j \leq m$ do not have some state with the same safely accepted words as $q_j$ in the SCC of $\mathcal{A}^B$ that $q^B$ is in (or the initial marking is different). 

So for each SCC of $\mathcal{A}$ we have the SCC is of minimal size among the ones equivalent in accepted language and markings, and as each such SCC is mapped to different equivalent SCCs in $\mathcal{A}^B$, we have that $\mathcal{A}$ is at most as big (in terms of the number of states) as $\mathcal{A}^B$. Since $\mathcal{A}^B$ is an arbitrary equivalent floating co-Büchi automaton (after potentially removing superfluous transitions between SCCs), this proves minimality.

Finally, to see that the procedure runs in polynomial time, note that each of the steps in the algorithm can be performed in polynomial time and the automata being processed never grows and can only shrink. Since in each repetition of all three steps, at least one transitions is removed from the automaton (or the algorithm terminates), and hence the number of iterations is limited by the number of transitions in the automaton, we obtain an overall polynomial computation time.
\end{proof}

\subsection{Proof of Lemma~\ref{lem:synchronizingWord}}

\begin{proof}
Let $w = w_0 \ldots w_h$ be a word such that $\delta(q,w)=q''$ for some state $q''$ that has a unique finite word $u$ that is only safely accepted from $q''$. Every SCC in a minimized automaton has such a $q''$/$u$ combination as otherwise the SCC would have been removed during minimization.
Since in a minimized floating co-Büchi automaton, the safe language of $q'$ can only be a subset of the safe language of $q$ if they are in the same SCC and every SCC has a state with a non-dominated safe language, also such a word $w_0 \ldots w_h$ has to exist.

Note that by the definition of safe languages, we have $\mathsf{Safe}(\delta(q,w_0 \ldots w_i)) \subseteq \mathsf{Safe}(\delta(q',w_0 \ldots w_i))$ for all $0 \leq i \leq h$. If for some such $0 \leq i \leq h$, we have $\mathsf{Safe}(\delta(q,w_0 \ldots w_i)) = \mathsf{Safe}(\delta(q',w_0 \ldots w_i))$, then the claim is proven, as then $w_0 \ldots w_i$ is a synchronizing word. 

Otherwise, we have that $w_0 \ldots w_h$ is a synchronizing word. Since $\delta(q,w_0 \ldots w_h u)$ is defined, and only from $q''$, $\delta(\cdot,u)$ is defined, we need to have that $\delta(q,w_0 \ldots w_h) = \delta(q',w_0 \ldots w_h) = q''$, as otherwise the safe language of $q'$ cannot be a superset of the safe language of $q$.
\end{proof}

\subsection{Proof of Lemma~\ref{lem:uniquenessLemma}}

\begin{proof}
We can construct such a word in a way similar to a part of the proof of Lemma~\ref{lem:minimizationLemma}.

We choose $q$ such that there is no other state in $Q'$ with a strictly larger language. Then, since states with the same language have been merged, there is some word $u$ with $\delta(q,u)$ being defined, but $\delta(q',u)$ being undefined for all $q' \in Q'$ with $q' \neq q$.

The word $u$ however does not satisfy the condition in the claim because there may be states in other SCCs of $\mathcal{A}$ that safely accept $u$.

Rather, we can build $w$ as follows:
Let $\{q_1, \ldots, q_n\}$ be the states of $\mathcal{A}$ with the same $m$-label and the same residual language as $q$ in $Q \setminus Q'$. 
Let furthermore $h : Q' \rightarrow \Sigma^*$ be a function mapping a state $q' \in Q'$ to some finite word $h(q')$ such that $\delta(q',h(q'))=q$, which exists since $Q'$ is the state set of a maximal SCC in $\mathcal{A}$ that contains $q$.
Let now $u_1, \ldots, u_n$ be a sequence of words safely accepted by $q$ but not by $q_1, \ldots, q_n$, respectively. As whenever there are states in different SCCs of a floating co-Büchi automaton with the same label and the same residual language with comparable safe languages, the minimization procedure removes one the states, we can assume that for the states $q_1, \ldots, q_n$, such finite words exists.
We define a function $g : Q' \rightarrow \Sigma^*$ such that $g(q_i) = u_i$ for all $1 \leq i \leq n$.

We now define a sequence of words $w^0 \ldots w^i$ with $w^0 = u h(\delta(q,u))$, and for all $1 \leq i \leq n$, we set:
\begin{equation*}
w^i = 
\begin{cases}
w^{i-1} g(\delta(q_i,w^{i-1})) h( \delta(& \\
\ \  q,w^{i-1} g(\delta(q_i,w^{i-1})))  & \text{if } 
\delta(q_i,w^{i-1}) \text{ is} \\
 & \text{defined}\\
w^{i-1} & \text{otherwise}.
\end{cases}
\end{equation*}
We have that $w = w^i$ satisfies the claim. The word is constructed to be safely accepted from $q$, and it is built in a way such that when the new word part added in $w^i$ starts for some $1 \leq i \leq n$, the word part that is not safely accepted from the state a run that was in $q_i$ before reading $w^{i-1}$ comes next. If such a run from $q_i$ already ended before, we have $w^i = w^{i-1}$, so nothing is added.
\end{proof}

\subsection{Proof of Lemma~\ref{lem:MLabeling}}

\begin{proof}
Note that assuming that $k$ is after the start of $\mathcal{A}^1$'s run comes without loss of generality because every point after a terminal rerailing point is by definition also a terminal rerailing point, so that it can always be placed after the start of the run of $\mathcal{A}^1$.

We perform the proof by contradiction. So assume the converse. Then among the mappings that satisfy only the second condition, there is a unique maximal one $m$ that we can obtain by for each $q^1 \in Q'^1$ taking a union of the sets $m'(q^1)$ for all such markings $m'$.

It has $m(q^1)=\emptyset$ for some $q^1 \in Q'^1$ (otherwise we already arrived at a contradiction as $m$ has the properties from the claim).
This means that for every $q \in Q$, there is a finite word $u^q$ such $\delta'^1(q^1,u^q)$ is defined, but a run in $\mathcal{R}$ starting from $q$ will lead outside $(\hat Q,\hat \delta)$ for $u^q$. 

We now stitch together a word $w$ such that $\delta'^1(q^1,w)$ is defined, but from no $q \in \hat Q$, every run from there for $w$ stays in $(\hat Q,\hat \delta)$. As $(\hat Q,\hat \delta)$ is the terminal SCC of a saturating word, which contains $w$ infinitely often, we then arrive at a contradiction.

Let $q_1, \ldots, q_n$ be some order of the states in $\hat Q$. We define a sequence of words $ w^0 \ldots  w^n$, set $ w^0 = \epsilon$, $ w =  w^n$, and for all $1 \leq i \leq n$, we define
\begin{equation*}
 w^i = \begin{cases}
 w^{i-1} & \text{if } \hat \delta^+(q_i, w^{i-1})=\emptyset \\
 w^{i-1} u^{ q'^i} v^i & \text{otherwise} \end{cases}
\end{equation*}
for some arbitrary $q'^i \in \hat\delta^+(q_i, w^{i-1})$ and some arbitrary finite word $v^i$ such that $q^1 = \delta'^1(q^1,u^{q'^i} v^i)$.

The word is built such that for each state $q_i$, for the word $ w^i$ there exists a suffix run that leaves $(\hat Q,\hat \delta)$. This is because either $w^{i-1}$ is already a word under which such a suffix run exists, or $\hat \delta^+(q_i, w^{i-1})$ is non-empty, and $w^{i-1}$ is followed by a finite word that leads to leaving $(\hat Q,\hat \delta)$ for some state in $\hat \delta^+(q_i, w^{i-1})$. Since for all $1 \leq i \leq n$, $w^{i-1}$ is a prefix of $w^{i}$, we ultimately have that $w$ is a word that from \emph{all} states in $\hat Q$ (each) induces a run leaving $(\hat Q,\hat \delta)$. At the same time, the choice of $v^i$ makes sure that a run in $\mathcal{A}^1$ from $q^1$ for $w$ stays in the maximal accepting SCC.
\end{proof}

\subsection{Proof of Lemma~\ref{lem:productPart}}
\begin{proof}

We solve this problem by reduction to parity game solving.
Given automata 
$\mathcal{A}^i = (Q^i,\Sigma,\delta^i,q^i_0)$,
$\mathcal{A}^{i+1} = (Q^{i+1},\Sigma,\delta^{i+1},q^{i+1}_0)$,
$\mathcal{A}^j = (Q^j,\Sigma,\delta^j,q^j_0)$ and
$\mathcal{A}^{j+1} = (Q^{j+1},\Sigma,\delta^{j+1},q^{j+1}_0)$, 
where for $\mathcal{A}^i$ and $\mathcal{A}^j$, we only consider the accepting  transitions in this part of the proof,
we construct a parity game $\mathcal{G} = (V^0,V^1,E^0,E^1,C,v^0)$ with the following structure:
\begingroup
\allowdisplaybreaks
\begin{align*}
V^0 & = Q^i \times Q^{i+1} \times Q^j \times Q^{j+1} \times \{0,1,2\} \\
V^1 & = Q^i \times Q^{i+1} \times Q^j \times Q^{j+1} \times \{0,1,2\} \times \Sigma \\
E^0 & = \{ ((q^i,q^{i+1},q^j,q^{j+1},z),(\delta^i(q^i,x),q^{i+1},\\
& \delta^j(q^j,x),q^{j+1},z',x))
\mid x \in \Sigma, z\neq 2 \rightarrow z'=z, \\
& z=2 \rightarrow z'=0
\}
\\
E^1 & = \{ ((q^i,q^{i+1},q^j,q^{j+1},z,x),(q^i,q'^{i+1},q^j,q'^{j+1}, \\
& z')) \mid \exists c^i,c^j:
(q^{i+1},x,q'^{i+1},c^i) \in \delta^{i+1} \wedge \\
& (q^{j+1},x,q'^{j+1},c^j) \in \delta^{j+1} \wedge (z=0 \rightarrow z'=2-c^i)\\ &
\wedge (z=1 \rightarrow z'=3-c^j) \wedge (z=2 \rightarrow z'=2)
\}
\\
C & = \Big\{ (q^i,q^{i+1},q^j,q^{j+1},z) \mapsto \begin{cases} 0 & \text{if } z=2 \\
1 & \text{otherwise,} \end{cases} \\
& \quad \quad (q^i,q^{i+1},q^j,q^{j+1},z,x) \mapsto 0
\Big\}\\
v^0 & \textit{ is arbitrary} \text{ (not used)}
\end{align*}
\endgroup

In this game, from some position $(q^i,q^{i+1},q^j,q^{j+1},z)$, the \emph{acceptance} player (player 0) tries to show that there is an infinite word word $\tilde w$ on which depending on the prefix words $w^1$ and $w^2$, the overall color of the word can either be $i$ or $j$, respectively, if we have that:
\begin{itemize} 
\item a run for $w^1$ can be in state $q^i$ in $\mathcal{A}^i$,
\item a run for $w^1$ can be in state $q^{i+1}$ in $\mathcal{A}^{i+1}$,
\item a run for $w^2$ can be in state $q^j$ in $\mathcal{A}^j$, and
\item a run for $w^2$ can be in state $q^{j+1}$ in $\mathcal{A}^{j+1}$.
\end{itemize}
Whether such prefix words exist is not checked when building the game. %

The $z$ component of the positions in the game is used to track if \emph{recently}, rejecting transitions have been taken for both $\mathcal{A}^{i+1}$ and $\mathcal{A}^{j+1}$. Whenever a rejecting transition for $\mathcal{A}^{i+1}$ is seen and the number is $0$, it flips to $1$ , and whenever a rejecting transition for $\mathcal{A}^{j+1}$ is seen and the number is $1$, it flips from $1$ to $2$. When the counter reaches a value of $2$, the play reaches a position with color $0$, and the $z$ component is reset to $0$.

The positions of the vertices encode state combinations in the considered automata. Player $0$ can choose next characters and in this way gradually build some word $\tilde w$, but can only do so in a way such that runs of $\mathcal{A}^i$ and $\mathcal{A}^j$ for $\tilde w$ are followed along. These runs can only take accepting transitions in $\mathcal{A}^i$ and $\mathcal{A}^j$, so that $\tilde w$ is accepted by the two. 
Player $1$ can choose how to resolve non-determinism in $\mathcal{A}^{i+1}$ and $\mathcal{A}^{j+1}$ for the word chosen by player $0$. 
The construction ensures that if and only if player $0$ has a strategy to win the parity game from some position $(q^i,q^{i+1},q^j,q^{j+1},z)$, then there exists a word $\tilde w$ with the properties mentioned above. This is because $\mathcal{A}^{i+1}$ and $\mathcal{A}^{j+1}$ are history-deterministic, so for each of them there exists a strategy to resolve non-determinism such that for every word, following the strategy yields an accepting run if and only if the respective automaton accepts the word. Since player $1$ can follow this strategy, if player $0$ then still wins, this means that both $\mathcal{A}^{i+1}$ and $\mathcal{A}^{j+1}$ reject the word because this means that along the play, the $z$ value switches from $0$ to $1$ infinitely often (indicating that a rejecting transition in $\mathcal{A}^{i+1}$ is taken), and the z value switches from $1$ to $2$ infinitely often (indicating that a rejecting transition in $\mathcal{A}^{j+1}$ is taken). So the existence of such a strategy for player $0$ from some position proves the existence of some word $\tilde w$ with the properties stated above. On the other hand, if such a word $\tilde w$ exists, then player $0$ can choose edges according to the letters of this word without the run for $\mathcal{A}^i$ and $\mathcal{A}^j$ leaving the game. As such a word is rejected by  $\mathcal{A}^{i+1}$ and $\mathcal{A}^{j+1}$, no matter which transitions player $1$ takes, the resulting play visits a vertex with color $0$ infinitely often. Then, player $0$ wins the play.

To obtain $R^{i,j}$ from the set of positions $P$ in $V^0$ that are winning for player $0$, we first make the following observation: if and only if 
\begin{itemize}
\item $w^1 \tilde w \in \mathcal{L}(\mathcal{A}^i)$, 
\item $w^1 \tilde w \notin \mathcal{L}(\mathcal{A}^{i+1})$, 
\item $w^2 \tilde w \in \mathcal{L}(\mathcal{A}^j)$, and
\item $w^2 \tilde w \notin \mathcal{L}(\mathcal{A}^{j+1})$,
\end{itemize}
then 
\begin{itemize}
\item there exists a state $q^i \in Q^i$ with $q^i \in \delta^{+,i}(q^i_0,w^1)$ and $\tilde w \in \mathcal{L}(\mathcal{A}^i_{q^i})$,
\item there exists a state $q^{i+1} \in Q^{i+1}$ with $q^{i+1} \in \delta^{+,i+1}(q^{i+1}_0,w^1)$ and $\tilde w \notin \mathcal{L}(\mathcal{A}^{i+1}_{q^{i+1}})$,
\item there exists a state $q^j \in Q^j$ with $q^{+,j} \in \delta^i(q^j_0,w^1)$ and $\tilde w \in \mathcal{L}(\mathcal{A}^j_{q^j})$, and
\item there exists a state $q^{j+1} \in Q^{j+1}$ with $q^{j+1} \in \delta^{+,j+1}(q^{j+1}_0,w^1)$ and $\tilde w \notin \mathcal{L}(\mathcal{A}^{j+1}_{q^{j+1}})$.
\end{itemize}
Furthermore, for $\mathcal{A}^i$ and $\mathcal{A}^j$ to accept the word, there has to exist a run for the word that eventually only takes accepting transitions. So $\tilde w$ can be split into $\tilde w = \tilde w_1 \tilde w_2$ such that 
\begin{itemize}
\item there exists a state $q'^i \in Q^i$ with $q'^i \in \delta^{+,i}(q^i_0,w^1 \tilde w^1)$ and there exists a run for $\tilde w^2$ from $q'^i$ in $\mathcal{A}^i$ only visiting accepting transitions, and 
\item there exists a state $q'^j \in Q^j$ with $q'^j \in \delta^{+,j}(q^j_0,w^2 \tilde w^1)$ and there exists a run for $\tilde w^2$ from $q'^j$ in $\mathcal{A}^j$ only visiting accepting transitions.
\end{itemize}
Let now $q'^{i+1}$ be some arbitrary state in $\delta^{+,i+1}(q^{i+1}_0,w^1 \tilde w^1)$ and $q'^{j+1}$ be some arbitrary state in $\delta^{+,j+1}(q^{j+1}_0,w^2 \tilde w^1)$.

By the description of the winning positions in the parity game above, we have that $(q'^i,q'^{i+1},q'^j,q'^{j+1}) \in P$ if and only if there exists some $\tilde w^2$ such that 
\begin{itemize}
\item $w^1 \tilde w^1 \tilde w^2 \in \mathcal{L}(\mathcal{A}^i)$, 
\item $w^1 \tilde w^1 \tilde w^2 \notin \mathcal{L}(\mathcal{A}^{i+1})$, 
\item $w^2 \tilde w^1 \tilde w^2 \in \mathcal{L}(\mathcal{A}^j)$, and
\item $w^2 \tilde w^1 \tilde w^2 \notin \mathcal{L}(\mathcal{A}^{j+1})$,
\end{itemize}

We now define $R^{i,j}$ as follows: we start with $R^{i,j} = \{(f^{L,i}(q'^i),f^{L,i+1}(q'^{i+1}),f^{L,j}(q'^j),f^{L,j+1}(q'^{j+1})) \mid \exists z \in \{0,1,2\}. (q'^i,q'^{i+1},q'^j,q'^{j+1},z) \in P\}$ for $f^{L,k}$ mapping states from $\mathcal{A}^k$ to their respective residual languages in $R^{L,k}$, respectively, for $0 \leq k \leq n$. By the above reasoning, this initial set $R^{i,j}$ only contains tuples that correspond to the claim.

In the next step, we saturate $R^{i,j}$ as follows: We gradually add tuples $(s''^i,s''^{i+1},s''^j,s''^{j+1})$ such that for some $x \in \Sigma$, we have $\delta^{L,i}(s''^i,x) = s'^i$, $\delta^{L,i+1}(s''^{i+1},x) = s'^{i+1}$, $\delta^{L,j}(s''^j,x) = s'^j$, and $\delta^{L,j+1}(s''^{j+1},x) = s'^{j+1}$ for some tuple $(s'^i,s'^{i+1},s'^j,s'^{j+1})$ already in $R^{i+1}$. Doing so is sound, because if we know that if for some words $w'^1$ and $w'^2$, there exists some word $\tilde w$ such that 
\begin{itemize}
\item $w'^1 x \tilde w \in \mathcal{L}(\mathcal{A}^i)$, 
\item $w'^1 x \tilde w \notin \mathcal{L}(\mathcal{A}^{i+1})$, 
\item $w'^2 x \tilde w \in \mathcal{L}(\mathcal{A}^j)$, and
\item $w'^2 x \tilde w \notin \mathcal{L}(\mathcal{A}^{j+1})$,
\end{itemize}
then we also have 
\begin{itemize}
\item $w'^1 \tilde w' \in \mathcal{L}(\mathcal{A}^i)$, 
\item $w'^1 \tilde w' \notin \mathcal{L}(\mathcal{A}^{i+1})$, 
\item $w'^2 \tilde w' \in \mathcal{L}(\mathcal{A}^j)$, and
\item $w'^2 \tilde w' \notin \mathcal{L}(\mathcal{A}^{j+1})$.
\end{itemize}
for some $\tilde w' \in \Sigma^\omega$.

By continuing the saturation process, we eventually add some tuple $(s'''^i,s'''^{i+1},s'''^j,s'''^{j+1})$ such that $s'''^i = \delta^{L,i}(s^{L,i}_0,w^1)$,
$s'''^{i+1} = \delta^{L,i+1}(s^{L,i+1}_0,w^1)$,
$s'''^j = \delta^{L,j}(s^{L,j}_0,w^2)$, and
$s'''^{j+1} = \delta^{L,j+1}(s^{L,j+1}_0,w^2)$, which ensures that $R^{i,j}$ has enough elements to satisfy the claim. 

Since along the way we ensured that only elements that can be in $R^{i,j}$ according to the claim have been added to $R^{i,j}$, the computed set $R^{i,j}$ is correct.

In terms of computation time, computing the initial version of $R^{i,j}$ can be done in polynomial time, as the two-color parity game has size polynomial in the input, and game solving for such games can be performed in polynomial time. Saturating $R^{i,j}$ also only takes polynomial time.
\end{proof}

\end{document}